\RequirePackage[dvipsnames]{xcolor}
\documentclass[11pt]{article}
\usepackage[margin=1in]{geometry}
\usepackage{graphicx}
\usepackage{microtype}
\usepackage{subcaption}
\usepackage{amsmath}
\usepackage{amssymb}
\usepackage{amsthm}
\usepackage{aliascnt}
\usepackage{ifthen}
\usepackage{tikz}
\usetikzlibrary{arrows.meta}
\usepackage[comma]{natbib}
\usepackage{todonotes}

\usepackage{nicefrac}

\newtheorem{theorem}{Theorem}[section]
\newaliascnt{proposition}{theorem}
\newtheorem{proposition}[proposition]{Proposition}
\aliascntresetthe{proposition}
\newaliascnt{lemma}{theorem}
\newtheorem{lemma}[lemma]{Lemma}
\aliascntresetthe{lemma}
\newaliascnt{corollary}{theorem}
\newtheorem{corollary}[corollary]{Corollary}
\aliascntresetthe{corollary}
\newaliascnt{fact}{theorem}

\aliascntresetthe{fact}
\theoremstyle{definition}
\newaliascnt{definition}{theorem}
\newtheorem{definition}[definition]{Definition}
\aliascntresetthe{definition}
\newaliascnt{example}{theorem}

\aliascntresetthe{example}
\theoremstyle{remark}
\newaliascnt{remark}{theorem}

\aliascntresetthe{remark}

\usepackage[hidelinks]{hyperref}
\usepackage[capitalize,nameinlink,noabbrev]{cleveref}
\crefname{proposition}{proposition}{propositions}
\crefname{lemma}{lemma}{lemmas}
\crefname{corollary}{corollary}{corollaries}
\crefname{fact}{fact}{facts}
\crefname{definition}{definition}{definitions}
\crefname{example}{example}{examples}
\crefname{remark}{remark}{remarks}

\newcommand{\floor}[1]{
{\lfloor {#1} \rfloor}
}

\newcommand{\ceiling}[1]{
{\lceil {#1} \rceil}
}

\DeclareMathOperator{\argmax}{argmax}

\newcommand{\given}{\,\mid\,}

\newcommand{\prob}[2][]{\text{\bf Pr}\ifthenelse{\not\equal{}{#1}}{_{#1}}{}\!\left[{\def\givenn{\middle|}#2}\right]}
\newcommand{\expect}[2][]{\text{\bf E}\ifthenelse{\not\equal{}{#1}}{_{#1}}{}\!\left[{\def\givenn{\middle|}#2}\right]}

\newcommand{\tparen}{\big}
\newcommand{\tprob}[2][]{\text{\bf Pr}\ifthenelse{\not\equal{}{#1}}{_{#1}}{}\tparen[{\def\given{\tparen|}#2}\tparen]}
\newcommand{\texpect}[2][]{\text{\bf E}\ifthenelse{\not\equal{}{#1}}{_{#1}}{}\tparen[{\def\given{\tparen|}#2}\tparen]}

\newcommand{\sprob}[2][]{\text{\bf Pr}\ifthenelse{\not\equal{}{#1}}{_{#1}}{}[#2]}
\newcommand{\sexpect}[2][]{\text{\bf E}\ifthenelse{\not\equal{}{#1}}{_{#1}}{}[#2]}

\newcommand{\dd}{{\mathrm d}}

 \usepackage{amsmath}
\usepackage{amssymb}
\usepackage{amsfonts}
\usepackage{ifthen}
\usepackage{xcolor}

\usepackage{bbm}
\usepackage{BOONDOX-cal}

\newcommand{\composed}[3]{#1{#2{#3}}}
\newcommand{\double}[2]{\kern.15ex#1{\kern-.15ex#1{\kern-.15ex#2}}}

\newcommand{\themodifier}{}

\newcommand{\ironed}{\bar}
\newcommand{\constrained}{\hat}

\newcommand{\thestage}{an}

\newcommand{\interimfont}[1]{\mathnormal{\color{black}#1}}

\newcommand{\interim}[1]{{\setstagefont{\interimfont}#1}}

\newcommand{\stagefont}{\mathrm}
\newcommand{\setstagefont}[1]{\renewcommand{\stagefont}{#1}}

\newcommand{\decoration}{\noaccents}

\newcommand{\stagemodifier}[2]{\composed {\ifthenelse{\equal{#2}{}}{}{\csname #2\endcsname}}{\ifthenelse{\equal{#1}{}}{}{\csname #1\endcsname}}}

\newcommand{\thestagemodifier}{\stagemodifier{\thestage}{\themodifier}}

\newcommand{\ensuredecoration}{\renewcommand{\decoration}{\thestagemodifier}}

\ensuredecoration

\newcommand{\optimized}[1]{#1\opt}
\newcommand{\differentiated}[1]{#1'}
\newcommand{\tagged}[2]{{#2}^{#1}}

\newcommand{\primedarg}[1]{#1\primed}
\newcommand{\noaccents}[1]{#1}

\newcommand{\additemsubscript}{}
\newcommand{\itemsubscript}{}

\newcommand{\drop}[1]{}
\newcommand{\definevariant}[3]{\expandafter\newcommand\expandafter{\csname \expandafter\drop\string#1#2\endcsname}{}
\expandafter\DeclareRobustCommand\expandafter{\csname \expandafter\drop\string#1#2\endcsname}{#3}}
\newcommand{\usevariant}[2]{\csname \expandafter\drop\string#1#2\endcsname}

\newcommand{\agind}[1][\agent]{_{\smash{#1\itemsubscript}}}
\newcommand{\agith}[1][\agent]{_{(\smash{#1})\smash{\itemsubscript}}}
\newcommand{\minusagind}[1][\agent]{_{\smash{-#1}}}

\newcommand{\newagentvar}[3][\decoration]{\definevariant{#2}{}{#1{\stagefont{#3}}\additemsubscript}\definevariant{#2}{i}{#1{\stagefont{#3}}\agind}
\definevariant{#2}{ith}{#1{\stagefont{#3}}\agith}
\definevariant{#2}{s}{\kern.1ex#1{\kern-.1ex\boldsymbol{\stagefont{#3}}}}
\definevariant{#2}{smi}{\kern.1ex#1{\kern-.1ex\boldsymbol{\stagefont{#3}}}\minusagind} 
}

\newcommand{\itind}[1][\itm]{_{#1}}
\newcommand{\minusitind}[1][\itm]{_{-#1}}

\newcommand{\newitemvar}[3][\decoration]{\definevariant{#2}{}{#1{\stagefont{#3}}}\definevariant{#2}{j}{#1{\stagefont{#3}}\itind}
\definevariant{#2}{s}{\kern.1ex#1{\kern-.1ex\boldsymbol{\stagefont{#3}}}}
\definevariant{#2}{smi}{\kern.1ex#1{\kern-.1ex\boldsymbol{\stagefont{#3}}}\minusitind} 
}

\newcommand{\newagentvarsm}[4]{\newagentvar[\stagemodifier{#3}{#4}]{#1}{#2}}

\newcommand{\theval}{v}
\newagentvarsm{\val}{\theval}{an}{none}
\newagentvarsm{\Val}{\theval}{interim}{forquant}
\newagentvarsm{\type}{t}{an}{none}
\newagentvarsm{\othertype}{s}{an}{none}

\newcommand{\therev}{R}
\newcommand{\thepricerev}{P}
\newagentvarsm{\rev}{\therev}{interim}{nomodifer}
\newagentvarsm{\marg}{\therev}{interim}{differentiated}
\newagentvarsm{\rawrev}{\thepricerev}{interim}{}
\newagentvarsm{\pricerev}{\thepricerev}{interim}{}
\newagentvarsm{\rawmarg}{\thepricerev}{interim}{differentiated}
\newagentvarsm{\pricemarg}{\thepricerev}{interim}{differentiated}

\newcommand{\thevirt}{\varphi}
\newcommand{\thecumvirt}{\Phi}
\newagentvarsm{\virt}{\thevirt}{interim}{}
\newagentvarsm{\cumvirt}{\thecumvirt}{interim}{}
\newagentvarsm{\qvirt}{\thevirt}{interim}{}
\newagentvarsm{\ivirt}{\thevirt}{interim}{ironed}
\newagentvarsm{\icumvirt}{\thecumvirt}{interim}{ironed}

\newagentvarsm{\dist}{F}{interim}{}
\newagentvarsm{\dens}{f}{interim}{}
\newagentvarsm{\hazard}{h}{interim}{}
\newagentvarsm{\cumhazard}{H}{interim}{}

\newcommand{\thealloc}{x}
\newagentvar{\alloc}{\thealloc}
\newagentvar[\ironed]{\ialloc}{\thealloc}

\newagentvarsm{\zee}{z}{an}{}
\newagentvarsm{\wye}{y}{an}{}
\newagentvarsm{\cee}{c}{an}{}
\newagentvarsm{\delt}{\delta}{an}{}

\newcommand{\thefalloc}{z}
\newagentvar{\falloc}{\thefalloc}
\newagentvarsm{\epfalloc}{\thefalloc}{expost}{}

\newcommand{\thequant}{q}
\newagentvarsm{\quant}{\thequant}{an}{none}
\newagentvarsm{\exquant}{\thequant}{an}{constrained}
\newagentvarsm{\critquant}{\thequant}{an}{constrained}  \newagentvarsm{\monoq}{\thequant}{an}{optconstrained}
\newagentvar[\interim]{\toquant}{\thequant}

\newcommand{\theqalloc}{y}
\newcommand{\thecumalloc}{Y}
\newagentvarsm{\qalloc}{\theqalloc}{interim}{none}
\newagentvarsm{\cumalloc}{\thecumalloc}{interim}{none}
\newagentvarsm{\calloc}{\theqalloc}{interim}{constrained}
\newagentvarsm{\cumcalloc}{\thecumalloc}{interim}{constrained}
\newagentvarsm{\iqalloc}{\theqalloc}{interim}{ironed}
\newagentvarsm{\icumalloc}{\thecumalloc}{interim}{ironed}

\newagentvar{\qprice}{\price}
\newagentvar{\qrev}{R}
\newagentvar[\ironed]{\iqrev}{\qrev}

\newagentvarsm{\excalloc}{\theqalloc}{expost}{constrained}
\newagentvarsm{\exalloc}{\theqalloc}{expost}{forquantile}
\newagentvarsm{\extalloc}{\thealloc}{expost}{forvalue}
\newagentvarsm{\exfalloc}{\thefalloc}{expost}{forvalue}

\newagentvarsm{\inducedrev}{\therev}{interim}{primedarg}

\newagentvar[\tilde]{\pseudorev}{\therev}
\newagentvar[\tilde]{\pseudorawrev}{\thepricerev}

\newagentvar{\typespace}{{\cal T}}
\newagentvar{\typesubspace}{S}

\newcommand{\theoutcome}{w}
\newagentvarsm{\outcome}{\theoutcome}{an}{none}
\newagentvarsm{\toutcome}{\theoutcome}{interim}{forval}
\newagentvar{\outcomespace}{{\cal W}}

\newcommand{\served}[1]{#1^1}
\newcommand{\nonserved}[1]{#1^0}
\newcommand{\alloced}[1]{#1^{\alloc}}
\newcommand{\allocedi}[1]{#1^{\alloci}}
\newagentvar[\alloced]{\xoutcome}{\outcome}
\newagentvar[\allocedi]{\xioutcome}{\outcome}
\newagentvar[\served]{\soutcome}{\outcome}

\newagentvar[\nonserved]{\nsoutcome}{\outcome}

\newcommand{\theprice}{p}
\newagentvar{\price}{\theprice}

\newcommand{\thepay}{p}
\newagentvarsm{\pay}{\thepay}{an}{}
\newagentvarsm{\vpay}{\thepay}{interim}{forval}
\newagentvarsm{\tpay}{\thepay}{interim}{forval}
\newagentvarsm{\epvpay}{\thepay}{expost}{forval}
\newagentvarsm{\talloc}{\thealloc}{interim}{forval}
\newagentvarsm{\valloc}{\thealloc}{interim}{forval}
\newagentvarsm{\eptalloc}{\thealloc}{expost}{forval}
\newagentvarsm{\epvalloc}{\thealloc}{expost}{forval}
\newagentvarsm{\epqalloc}{\theqalloc}{expost}{forquant}

\newcommand{\theaalloc}{y}
\newagentvarsm{\aalloc}{\theaalloc}{interim}{forbid}
\newagentvarsm{\epaalloc}{\theaalloc}{expost}{forbid}
\newagentvarsm{\apay}{\thepay}{interim}{forbid}
\newagentvarsm{\epapay}{\thepay}{expost}{forbid}

\newagentvarsm{\balloc}{\thealloc}{interim}{forbid}
\newagentvarsm{\bpay}{\thepay}{interim}{forbid}
\newagentvarsm{\epballoc}{\thealloc}{expost}{forbid}
\newagentvarsm{\epbpay}{\thepay}{expost}{forbid}

\newagentvar{\act}{a}
\newagentvar{\bidspace}{A}
\newagentvar{\actspace}{A}

\newagentvarsm{\critval}{\theval}{an}{constrained}
\newagentvarsm{\epvcritval}{\theval}{expost}{constrained}

\newagentvar[\constrained]{\crittype}{\type}
\newagentvar[\constrained]{\critvirt}{\virt}
\newagentvar[\constrained]{\reserve}{\val} \newagentvar[\constrained]{\bidreserve}{\bid} \newagentvarsm{\monop}{\val}{an}{optconstrained}
\newagentvar[\constrained]{\monot}{\type}
\newagentvar[\optimized]{\monorev}{\rev}

\newcommand{\thercalloc}{y}
\newagentvarsm{\rcalloc}{\thercalloc}{an}{none}
\newagentvarsm{\optrcalloc}{\thercalloc}{expost}{optimized}
\newagentvarsm{\biddist}{G}{interim}{none}
\newagentvarsm{\biddens}{g}{interim}{none}
\newagentvarsm{\critbid}{\bid}{an}{constrained}
\newagentvarsm{\cbid}{B}{an}{constrained}
\newagentvar[\primedarg]{\wbid}{\bid}
\newagentvar{\gfunc}{\vartheta}
\newagentvar{\block}{C}

\newagentvarsm{\pee}{p}{interim}{none}

\newagentvarsm{\util}{u}{an}{}
\newagentvarsm{\vutil}{u}{interim}{forvalue}
\newagentvarsm{\butil}{u}{interim}{forbid}

\newcommand{\thebid}{b}
\newagentvar{\bid}{\thebid}
\newagentvarsm{\strat}{\thebid}{interim}{}

\newagentvar[\tagged{\text{SD}}]{\sdvirt}{\virt}
\newagentvar[\composed{\tagged{\text{SD}}}{\ironed}]{\sdivirt}{\virt}
\newagentvar[\tagged{\text{MD}}]{\mdvirt}{\virt}
\newagentvar[\composed{\tagged{\text{MD}}}{\ironed}]{\mdivirt}{\virt}

\newagentvar[\ironed]{\iprice}{\price}  \newagentvar[\ironed]{\ival}{\val}  

\newagentvar{\ints}{{\cal I}}

\newcommand{\feasibles}{{\cal X}}

\newagentvar{\wal}{w}

\newitemvar{\pos}{j}
\newitemvar{\weight}{w}
\newitemvar[\differentiated]{\mweight}{\weight}
\newitemvar{\udtype}{\type}
\newitemvar[\constrained]{\udcrittype}{\type}
\newitemvar{\udalloc}{\alloc}
\newitemvar{\udprice}{\price}
\newitemvar[\constrained]{\udcalloc}{\qalloc}
\newitemvar[\constrained]{\udcumcalloc}{\cumalloc}

\newagentvar{\mech}{M}
\newagentvar[\skew{5}{\hat}]{\cmech}{M}
\newagentvar{\alg}{{\cal A}}

\newcommand{\Rev}[1]{\operatorname{Rev}\!\left(#1\right)}

\newagentvarsm{\budget}{B}{an}{none}

\newagentvarsm{\rawprice}{\theprice}{expost}{forbid}
\newagentvarsm{\rawalloc}{\thealloc}{expost}{forbid}

\newcommand{\thecost}{c}
\newagentvarsm{\cost}{\thecost}{an}{none}
\newagentvarsm{\qcost}{\thecost}{interim}{forquant}

\DeclareMathOperator{\OPT}{OPT}

\DeclareMathOperator{\WM}{WM}
\DeclareMathOperator{\RM}{RM}

\newagentvarsm{\trans}{\sigma}{interim}{none}

\newagentvar{\demandset}{S}

\newcommand{\opt}{^{\star}}
\newcommand{\primed}{^\dagger}

\newagentvar{distout}{w}
\newagentvar{idistout}{\bar{w}}

\newcommand{\uniformzeroone}{U[0,1]}

\newcommand{\positivepart}[1]{\smash{(#1)^+}}

\newcommand{\apx}{\alpha}

\newagentvarsm{\epricerev}{\thepricerev}{an}{}
\newagentvarsm{\erev}{\therev}{an}{}
\newagentvarsm{\emarg}{\therev}{an}{differentiated}
\newagentvarsm{\epricemarg}{\thepricerev}{an}{differentiated}
\newagentvarsm{\efalloc}{\theqalloc}{interim}{none}
\newagentvarsm{\evirt}{\thevirt}{an}{none}

\newcommand{\maxval}{h}

\DeclareMathOperator{\RP}{RP}

\newagentvar{\gap}{\delta}

\newagentvar[\tilde]{\ualloc}{\alloc}
\newagentvar[\tilde]{\uprice}{\price}

\renewcommand{\mech}{M}
\newcommand{\Mechs}{\mathbb{M}}
\newcommand{\Dists}{\mathbb{F}}
\DeclareMathOperator{\RPP}{RPP}
\newcommand{\revM}{\Mechs_{\mathrm{rev}}}
\newcommand{\allM}{\Mechs_{\mathrm{all}}}

\newagentvar{\bundle}{S}
\newcommand{\pricedist}{G}
\newcommand{\pointmass}[1]{\disti[#1]}
\newcommand{\utilB}{\util_B}
\newcommand{\utilS}{\util_S}

\DeclareMathOperator{\Wel}{Wel}
\DeclareMathOperator{\FB}{FB}
\DeclareMathOperator{\SB}{SB}

\DeclareMathOperator{\MHR}{MHR}

\newcommand{\apxratio}{\rho}

\newcommand{\CompEff}{\mu}
\newcommand{\IndEff}{\eta}
\newcommand{\regretvalue}{\delta}

\newcommand{\OptPIratio}{\rho}
\newcommand{\RevGap}[1]{1-\frac{\OptPIratio(\revM,#1)}{\OptPIratio(\allM,#1)}}
\newagentvar{\sample}{s}
\newcommand{\samplemult}{a}

\newcommand{\sellerarea}{S}
\newcommand{\buyerarea}{B}

\title{Robustness in Mechanism Design \thanks{This is a preprint of a
    review article that is accepted to appear in the {\em Annual Review of
    Economics} in August of 2027.  It contains an additional section on
    the revelation gap that, for space reasons, was mostly removed
    from the published version.  Thanks to Yingkai Li, Simon
    Loertscher, Larry Samuelson, and Kangning Wang for their
    feedback.}}  \author{Jason Hartline \thanks{Department
    of Computer Science, Northwestern University, Evanston, Illinois,
    USA; \url{hartline@northwestern.edu}.}} \date{}

\begin{document}

\maketitle

\begin{abstract}
The article reviews the robust design and analysis of auctions,
focusing on the predominant paradigm in the computer science
literature, which quantifies robustness by the worst-case, over
environments, ratio between the robust auction's performance and that
of an optimal auction for the environment.  Applications include robust
analysis of canonical auctions such as the first-price auction for
the welfare objective, prior-robust auction design for the
revenue objective, and prior-robust design of exchange mechanisms for the
gains-from-trade objective.  Finally, the revelation principle is not
without loss in robust mechanism design.  The potential loss can be
quantified by the revelation gap, the fraction of the optimal robust
guarantee that is lost when restricting from all mechanisms to
revelation mechanisms.

 \end{abstract}


\section{Introduction}

At the turn of the century, when theoretical computer scientists began
to study algorithms as mechanisms, due primarily to concerns of
strategic behavior in distributed systems such as the Internet, they
adopted the prevailing analysis paradigm from approximation algorithms
\citep[see the textbook,][]{vaz-01}.  An algorithm's performance should be
measured by the worst-case fraction it obtains of the optimal
performance in the same scenario.  This framework was initially applied to
robust mechanism design in \citet{GHW-01} and
\citet{NR-01} and to robust analysis of equilibria in \citet{KP-99,KP-09} and
\citet{RT-02}.\footnote{The first paper was recognized by the 2021 ACM
SIGecom Test of Time Award and the latter three were recognized by the
2012 G{\"o}del Prize for their significant impact in computer
science.}  There are two key aspects of this analysis framework: it is
a robustness framework and it is scale invariant.

Robustness is a critical concern in computer science.  Algorithms are
often designed for reuse across heterogeneous applications rather than
tailored to a particular input environment.  For example, the same
sorting algorithm may be called as a subroutine by many programs on
very different inputs.  Computer science therefore values algorithms
that provide useful guarantees without being redesigned for each
application.

Scalability is a critical concern in computer science.  For example,
the canonical notion of polynomial runtime as the definition of
computational tractability comes from the fact that it implies that
performance scales well.  If the runtime of an algorithm on an input
of length $n$ is a polynomial, like $T(n) = n^k$ for a constant $k$,
then when the input gets a constant factor $c$ larger, the runtime
gets only a constant factor $c^k$ longer, as $T(c n) = c^k n^k = c^k
T(n)$.  Scaling well is important in computer systems.  For example,
as a computer network grows, it would be problematic if the algorithms
that support it needed to be redesigned.  For example, as a
technological startup company increases its user base, it would be
problematic if key technological infrastructure needed to be
redesigned.

Robustness to scale asks that an algorithm be good both on
small inputs and on large inputs.  Translating this concern to the
economic analysis of auctions, it asks that the performance be good
when the optimal performance is small and that the performance be good
when the optimal performance is large.

In contrast, there are two robustness paradigms that are popular in
the economics literature on auctions, max-min absolute performance and
min-max regret.  The former asks for a mechanism that maximizes its
minimum performance over environments \citep[e.g.,][]{BS-11}; the
latter asks for a mechanism that minimizes its worst-case regret,
i.e., the difference between the optimal performance in the
environment and its performance \citep[e.g.,][]{BS-08}.  When the goal
is robustness to scale, these latter two robustness paradigms give
only trivial answers.  As the scale goes to zero, max-min absolute
performance goes to zero; as the scale goes to infinity, min-max
regret becomes unbounded.  See the detailed example and comparison
in
\Cref{s:robustness-paradigms}.

One economic justification for evaluating robustness as a ratio comes from
comparing it to the standard multiplicative commission structures for
intermediaries \citep{HR-14}.  For example, online marketplaces like
eBay and Booking.com and traditional intermediaries like real estate
agents charge sellers a multiplicative fee on the sale price.  Thus, a
principal deciding whether to adopt such a platform or intermediary evaluates their
cost multiplicatively.  Likewise, it is natural to view a robust
auction as an intermediary that need not do the optimal
thing in every environment but is convenient to use and to evaluate
the non-optimality of the robust mechanism multiplicatively.

For example, since max-min ratio allows analysis of auctions across
scales, it allows for answering of the unqualified question: what
auctions are good?

The position taken in this review article is that when robustness to
scale is desired, the appropriate robustness paradigm is max-min
ratio.  The other robustness paradigms have their applications.  For
max-min absolute performance, for example, \citet{car-17} shows that a
multi-item seller who knows an additive buyer's marginal distributions
for the items, but not the correlation structure, maximizes worst-case
revenue by selling the items separately.
\citet{BBM-17} give a max-min absolute analysis of revenue in a
first-price auction over possible information structures of the
buyers.  \citet{BD-21} characterize the optimal auction for max-min
absolute revenue across information structures and equilibria.  The
min-max regret notion is common in statistical problems such as learning in games \citep[see review article,][]{har-27} and
robust information aggregation, e.g., \citet{ABS-18}.  Notably,
\citet{GHHK+25} show that min-max regret gives a better understanding of
robust aggregation than either max-min ratio or max-min absolute
performance.  Importantly, none of these environments allow arbitrary
scaling of the optimal performance.  See the survey by \citet{car-19}
for an extensive review of these literatures.

This review article will examine in depth a few especially salient
topics in the recent computer science literature on robust mechanism
design.  \Cref{s:robust-analysis-auction-equilibria}
develops robust methods for analyzing the welfare in winner-pays-bid
auctions.  These methods are robust to value distributions,
information structures, and equilibrium selection; moreover, they
decompose the loss into easy-to-analyze contributions from aggregate
uncertainty and unpriced
externalities.  \Cref{s:robust-auction-design} considers revenue
maximization in auctions with i.i.d.\ buyers.  It shows that the
second-price auction guarantees a $1-\nicefrac1n$ fraction of the optimal
revenue, but some other auction formats obtain a higher
fraction of the optimal revenue.  \Cref{s:robust-exchange} develops
robust mechanisms for optimizing gains from trade in bilateral
exchange.  We see that the random proposer mechanism, where either the
buyer or seller is picked at random to propose a take-it-or-leave-it
price to the other agent, obtains at least a $\nicefrac12$ fraction of the
second-best gains from trade and at least a $\nicefrac14$ fraction of the
first-best gains from trade.  Subsequent work improves the first-best
guarantee to $\nicefrac1e\approx 0.37$ \citep{BWZ-26}.  This mechanism obtains these robust
guarantees in every perfect Bayesian equilibrium.
\Cref{s:revelation-gap} begins
with the observation that the revelation principle does not hold for
robust mechanism design, and asks to quantify the robust performance gap
between revelation mechanisms and all mechanisms, including non-revelation
mechanisms.  For example, the random-proposer mechanism for bilateral
exchange obtains $\nicefrac12$ of the second-best gains from trade, but no
revelation mechanism obtains any constant fraction of it.  Thus,
bilateral exchange exhibits a total revelation gap, i.e., one.

This review article, to maintain a consistent framework throughout,
has presented prior-robust analyses as max-min ratios of a robust
mechanism's performance against the second-best performance.  These
ratios are numbers in $[0,1]$ where larger is better.  Many of the
original sources analyze instead the reciprocal quantity, i.e., the
min-max ratio of the second-best performance to the mechanism's
performance.  This is a matter of convention only, and does not change
any qualitative results.  In keeping with the $[0,1]$ convention, the
revelation gap is the fraction of the optimal all-mechanism max-min guarantee
that is lost when attention is restricted to revelation mechanisms.  This is
also a $[0,1]$ ratio with zero representing no gap and one representing a total gap
where there is no non-trivial max-min ratio guarantee for revelation
mechanisms.

\section{Robustness Paradigms}
\label{s:robustness-paradigms}

Robust mechanism design begins with a choice of robustness paradigm.
For the same family of value distributions, max-min absolute
performance, max-min ratio, and min-max regret generally select
different mechanisms because they protect against different losses.
The illustration presented here on the monopolist pricing problem is
originally an example from \citet{HJL-25}.

While the economics literature focuses on max-min absolute performance
and min-max regret, this survey showcases the max-min ratio.  The
analysis below makes the distinction concrete: the
max-min absolute performance is governed by low-revenue settings where
there is little to gain, the min-max regret by high-revenue settings
where there is much to lose, and the max-min ratio by a proportional
guarantee at every scale.

\subsection{Robustness Paradigms}

Let \(\Dists\) be a family of value distributions and \(\Mechs\) a
family of mechanisms.  Let \(\vals\) denote a value profile and
\(\dist\) a value distribution.  For mechanism
\(\mech\in\Mechs\), write \(\mech(\vals)\) for its performance on
realized values \(\vals\), and define its expected performance under distribution
\(\dist\in\Dists\) by
$\mech(\dist)=\expect[\vals\sim\dist]{\mech(\vals)}$.  For each
\(\dist\in\Dists\), let \(\SB_{\dist} = \argmax_{\mech \in \Mechs}
\mech(\dist)\) denote a second-best mechanism tailored to \(\dist\),
and write \(\SB(\dist)=\SB_{\dist}(\dist)\) for its expected
performance.\footnote{The term second best ($\SB$)
contrasts with first best ($\FB$) to distinguish mechanism-design
problems with incentive and participation constraints from optimization
problems without those constraints. Generally, this
review article compares prior-robust mechanisms to second best to isolate the difficulty of robustness from that of incentive and participation constraints.}

\begin{figure}[t]
\centering \begin{subfigure}[t]{0.31\linewidth}
\centering
\begin{tikzpicture}[xscale=2.82,yscale=0.92,every node/.style={font=\scriptsize}]
\fill[gray!35] (0.2,0.2) -- (1,0.2) -- (1,2) -- (0.2,0.4) -- cycle;
\node at (.45,1.6) {$\SB(\dist)$};
\draw[dotted] (0,0) -- (0.2,0.4);
\draw[dotted] (1,2) -- (1.2,2.4);
\draw[dashed] (0.2,0.2) -- (1.2,0.2);
\draw[thick] (0.2,0.4) -- (1,2);
\node at (.70,0.80) {$\mech(\dist)$};
\draw (0,-0.05) -- (0,0.05);
\draw[<->] (0,2.4) -- (0,0) -- (1.2,0);
\end{tikzpicture}
\caption{Max-min absolute.}
\label{fig:robustness-objectives-absolute}
\end{subfigure}\hfill
\begin{subfigure}[t]{0.31\linewidth}
\centering
\begin{tikzpicture}[xscale=2.82,yscale=0.92,every node/.style={font=\scriptsize}]
\fill[gray!35] (0.2,0.2) -- (1,1) -- (1,2) -- (0.2,0.4) -- cycle;
\node at (.4,1.6) {$\SB(\dist)$};
\draw[dotted] (0,0) -- (0.2,0.4);
\draw[dotted] (1,2) -- (1.2,2.4);
\draw[dashed] (0,0) -- (1.2,1.2);
\draw[thick] (0.2,0.4) -- (1,2);
\node at (.69,1.03) {$\mech(\dist)$};
\draw (0,-0.05) -- (0,0.05);
\draw[<->] (0,2.4) -- (0,0) -- (1.2,0);
\end{tikzpicture}
\caption{Max-min ratio.}
\label{fig:robustness-objectives-ratio}
\end{subfigure}\hfill
\begin{subfigure}[t]{0.31\linewidth}
\centering
\begin{tikzpicture}[xscale=2.82,yscale=0.92,every node/.style={font=\scriptsize}]
\fill[gray!35] (0.2,0) -- (0.4,0) -- (1,1.2) -- (1,2) -- (0.2,0.4) -- cycle;
\node at (.4,1.6) {$\SB(\dist)$};
\draw[dotted] (0,0) -- (0.2,0.4);
\draw[dotted] (1,2) -- (1.2,2.4);
\draw[dashed] (0.4,0) -- (1.2,1.6);
\draw[thick] (0.2,0.4) -- (1,2);
\node at (.62,0.86) {$\mech(\dist)$};
\draw (0,-0.05) -- (0,0.05);
\draw[<->] (0,2.4) -- (0,0) -- (1.2,0);
\end{tikzpicture}
\caption{Min-max regret.}
\label{fig:robustness-objectives-regret}
\end{subfigure}
 \caption{Three robustness paradigms impose different requirements.
The horizontal axis orders distributions by second-best performance
\(\SB(\dist)\), the thick line is the distribution-specific optimum,
and \(\mech(\dist)\) is the performance of a fixed mechanism.  In
\subref{fig:robustness-objectives-absolute}, max-min absolute performance
protects an absolute floor; in \subref{fig:robustness-objectives-ratio},
max-min ratio scales with \(\SB(\dist)\); and in
\subref{fig:robustness-objectives-regret}, min-max regret allows a fixed
additive loss.
Adapted from \citet{HJL-25}.}
\label{fig:robustness-objectives-schematic}
\end{figure}

\begin{definition}[Robustness paradigms]
\label{def:robustness-paradigms}
Given a family \(\Mechs\) of mechanisms and a family \(\Dists\) of value
distributions, the three robustness paradigms, illustrated in
\Cref{fig:robustness-objectives-schematic}, are:
\[
\begin{aligned}
  \text{max-min absolute performance:}\quad
    &\sup\nolimits_{\mech\in\Mechs}\inf\nolimits_{\dist\in\Dists}\mech(\dist), \\
  \text{max-min ratio:}\quad
    &\sup\nolimits_{\mech\in\Mechs}\inf\nolimits_{\dist\in\Dists}
      \frac{\mech(\dist)}{\SB(\dist)}, \\
  \text{min-max regret:}\quad
    &\inf\nolimits_{\mech\in\Mechs}\sup\nolimits_{\dist\in\Dists}
      \bigl(\SB(\dist)-\mech(\dist)\bigr).
\end{aligned}
\]
For a fixed mechanism, the three inner criteria are its minimum absolute
performance, minimum ratio, and maximum regret.  A ratio with
\(\SB(\dist)=0\) is one.
\end{definition}

\subsection{Example: Monopoly Pricing}
\label{ex:robust-monopoly-comparison}

Consider a seller aiming to robustly maximize revenue in the sale of a
single item to a single buyer.  The buyer's value \(\val\) is drawn
from an unknown distribution \(\dist\) supported on \([1,\maxval]\)
with \(\maxval \geq e\).\footnote{The assumption \(\maxval \geq e\)
avoids a lower-support boundary case in the min-max regret analysis;
see \citet{BS-08} for the more general construction and analysis.}
Here \(\Dists\) is the family of all such distributions and \(\Mechs\)
is the family of randomized posted-price mechanisms.\footnote{The
restriction to randomized posted-price mechanisms is without loss for
the single-agent problem: \Cref{prop:all-prior-ex-post} reduces the interim 
constraints to their ex-post counterparts, and the resulting mechanism
can be represented as a distribution over posted prices
\citep{mye-81}.}  A randomized posted-price mechanism
\(\RPP_{\pricedist}\) is identified by a distribution \(\pricedist\)
over prices.  The revenue from posting price \(\critval\) is
\(\critval\, \prob[\val \sim \dist]{\val \geq \critval}\).  By
\citet{mye-81}, the optimal mechanism for distribution \(\dist\), denoted
\(\SB_{\dist}\), posts the monopoly price
$\monop=\argmax_{\critval}
\critval\,\prob[\val\sim\dist]{\val\geq\critval}$.

We now solve the robust monopoly-pricing example under each paradigm.
Each problem can be viewed as a zero-sum game between the mechanism
designer and nature.  It suffices to (a) identify the value of each
game and strategies for (b) the mechanism designer and (c) nature that
guarantee the value.  The analysis will also (d) evaluate the
performance of the designer's mechanism in the two other robustness
paradigms.  For each of these games, it suffices to consider nature's
mixed strategy over point-mass value distributions; we denote the
point mass at \(\val\) by \(\pointmass{\val}\), with
\(\SB(\pointmass{\val})=\val\).  The performance of the optimal
mechanism for each paradigm is compared across the three paradigms in
\Cref{tab:robust-monopoly-objectives}.

\begin{table}
\centering
\caption{Comparison of the robustness paradigms in the
  monopoly-pricing example.  Rows correspond to the optimal mechanism
  in each paradigm; columns evaluate that mechanism in each paradigm.
  The optimal for each is boxed.  Adapted from \citet{HJL-25}.}
\label{tab:robust-monopoly-objectives}
\begin{tabular}{lccc}
\hline
mechanism & minimum revenue & minimum ratio & maximum regret \\
\hline
max-min absolute & \fbox{\(1\)} & \(\frac{1}{\maxval}\) & \(\maxval-1\) \\
max-min ratio & \(\frac{1}{1+\ln \maxval}\) &
  \fbox{\(\frac{1}{1+\ln \maxval}\)} &
  \(\maxval-\frac{\maxval}{1+\ln \maxval}\) \\
min-max regret & \(0\) & \(0\) & \fbox{\(\frac{\maxval}{e}\)} \\
\hline
\end{tabular}
\end{table}

\begin{proposition}[\citealp{HJL-25}]
\label{prop:max-min-absolute-monopoly-pricing}
In the robust monopoly-pricing example, the randomized posted-pricing
rule selected by max-min absolute performance posts price \(1\).  Its
minimum revenue is \(1\), its minimum ratio is \(\nicefrac{1}{\maxval}\), and its
maximum regret is \(\maxval-1\).
\end{proposition}

\begin{proof}
(a) In this game between the mechanism designer and nature, the payoff is
  the absolute performance.  The value of the game is \(1\).  We verify this
  value by giving strategies for the mechanism designer and nature that
  guarantee a payoff of \(1\).

  (b) The mechanism designer guarantees a payoff of one by posting a price
  of $\critval = 1$.

  (c) Nature guarantees a payoff of at most one by playing the value
  distribution $\pointmass{1}$.

  (d) Against \(\pointmass{\val}\), the tailored optimum is \(\val\),
  while the revenue from posting price \(1\) is \(1\).  The ratio is
  therefore \(\nicefrac{1}{\val}\), and the regret is \(\val-1\).  The worst case
  is $\val = \maxval$.
\end{proof}

\begin{proposition}[\citealp{BS-08,HJL-25}]
\label{prop:min-max-regret-monopoly-pricing}
In the robust monopoly-pricing example, the randomized posted-pricing rule
selected by min-max regret has CDF
\[
  \pricedist(\critval)=
  \begin{cases}
    0, & \critval\in[1,\nicefrac{\maxval}{e}),\\
    1+\ln(\critval/\maxval), & \critval\in[\nicefrac{\maxval}{e},\maxval].
  \end{cases}
\]
Its minimum revenue is \(0\), its minimum ratio is \(0\), and its maximum regret
is \(\nicefrac{\maxval}{e}\).
\end{proposition}

\begin{proof}
(a) In this game between the mechanism designer and nature, the payoff is
  the regret.  The value of the game is
  \(\regretvalue=\nicefrac{\maxval}{e}\).  We verify this value by giving strategies for
  the mechanism designer and nature that guarantee a payoff of
  \(\regretvalue\).

(b) The designer's pricing rule has density \(\nicefrac{1}{\critval}\) on
\(\critval \in [\regretvalue,\maxval]\).  Against a point mass at
\(\val<\regretvalue\), its revenue is zero.  For
\(\val\ge\regretvalue\), its revenue is
\(\int_{\regretvalue}^{\val}\critval\,(\nicefrac{1}{\critval})\,\dd\critval
=\val-\regretvalue\).  The payoff is therefore
\(\val-\positivepart{\val-\regretvalue}=\min\{\val,\regretvalue\}\), which is
at most \(\regretvalue\) for every \(\val\).

(c) Nature plays a distribution over \(\pointmass{\val}\), with \(\val\)
drawn from the truncated equal-revenue distribution \(Q\) on
$[\regretvalue,\maxval]$, with CDF \(Q(\val)=1-\regretvalue/\val\)
for \(\val \in [\regretvalue,\maxval)\) and a point mass at $\maxval$ of
\(\regretvalue/\maxval\).  Every price
\(\critval\in[\regretvalue,\maxval]\) earns revenue
\(\critval\,\prob[\val\sim Q]{\val\ge\critval}
=\critval\, \regretvalue / \critval=\regretvalue\), while every price
below \(\regretvalue\) earns less than \(\regretvalue\).  The payoff at
each point mass is the welfare minus the revenue.  The expected
welfare is \(\expect[\val\sim Q]{\val}
=\regretvalue+\int_{\regretvalue}^{\maxval}(\regretvalue/\val)\dd\val
=2\regretvalue\); thus, nature guarantees a payoff of
\(2\regretvalue-\regretvalue=\regretvalue\).

(d) The designer's pricing distribution earns $0$ against point masses
at or below \(\regretvalue\), so its minimum absolute revenue and minimum
ratio are both $0$.
\end{proof}

\begin{proposition}[\citealp{HR-14,HJL-25}]
\label{prop:max-min-ratio-monopoly-pricing}
In the robust monopoly-pricing example, the randomized posted-pricing rule
selected by max-min ratio has CDF
\[
  \pricedist(\critval)=\frac{1+\ln \critval}{1+\ln \maxval},
  \qquad \critval\in[1,\maxval].
\]
Its minimum revenue and minimum ratio are both
\(1/(1+\ln \maxval)\), and its maximum regret is
\(\maxval-\maxval/(1+\ln \maxval)\).
\end{proposition}

\begin{proof}
(a) In this game between the mechanism designer and nature, the payoff is
  the ratio.  The value of the game is
\(\apxratio=1/(1+\ln \maxval)\).  We verify this value by giving
strategies for the mechanism designer and nature that guarantee a payoff
of \(\apxratio\).

(b) The designer's pricing rule places probability \(\apxratio\) on price \(1\)
and has density \(\nicefrac {\apxratio}{\critval}\) above \(1\).  Against a point mass at
\(\val\), its revenue is
\(\apxratio+\int_1^{\val}\critval\,(\nicefrac{\apxratio}{\critval})\,\dd\critval
=\apxratio\,\val\).  Thus, the payoff is \(\apxratio\) for every \(\val\).

(c) Nature plays \(\pointmass{\val}\), where \(\val\) is drawn from a
distribution with density \(\nicefrac{\apxratio}{\val}\) on \([1,\maxval)\) and mass
\(\apxratio\) at \(\maxval\).  Against \(\pointmass{\val}\), price
\(\critval\) has payoff \(\nicefrac{\critval}{\val}\) when
\(\val\ge\critval\).  Thus, every price
\(\critval\in[1,\maxval]\) gives expected payoff
\[
\int_{\critval}^{\maxval}
\frac{\critval}{\val}\,\frac{\apxratio}{\val}\,\dd\val
+\apxratio\,\frac{\critval}{\maxval}
=\apxratio\,\critval
\left(\frac{1}{\critval}-\frac{1}{\maxval}\right)
+\apxratio\,\frac{\critval}{\maxval}=\apxratio.
\]
Thus, nature guarantees a payoff of \(\apxratio\).

(d) Evaluating the designer's price distribution according to minimum
absolute revenue and maximum regret gives \(\apxratio\), attained at
\(\val=1\), and \((1-\apxratio)\,\maxval\), attained at \(\val=\maxval\),
respectively.
\end{proof}

\subsection{Interpretations}

The three paradigms differ in which distributions are worst cases.
Max-min absolute performance has its worst case where there is not much to
gain, i.e., distribution $\pointmass{\val}$ with $\val = 1$.  Min-max regret obtains
its worst performance where there is the most to lose, i.e.,
distribution $\pointmass{\val}$ with $\val=\maxval$.  Max-min ratio
obtains the same worst case everywhere on $[1,\maxval]$.

As the range \([1,\maxval]\) expands, the endpoint normalizations used
by max-min absolute performance and min-max regret become increasingly
concentrated near opposite ends of the range.  The applications highlighted in
the remainder of this survey place no positive lower bound or finite upper
bound on scale.  Scaling values toward zero forces max-min absolute performance
to zero, while scaling them toward infinity makes min-max regret unbounded.
The ratio criterion is therefore the only one of the three paradigms that gives
a nontrivial guarantee.  This is the sense in which the
max-min ratio paradigm is robust to scale \citep{HJL-25}.

A standard critique of performance ratios in economics is that
game-theoretic behavior is invariant to adding a constant to a
player's payoffs, whereas a ratio depends on the location of zero \citep[cf.][]{GS-24}.  In
the applications considered here, zero has a canonical interpretation
as the status quo: no allocation in the welfare mazimization, no revenue
in revenue maximization, and no trade in bilateral exchange.  This status quo
anchors the ratio's normalization.  

The next three sections apply this paradigm to three
objectives.  \Cref{s:robust-analysis-auction-equilibria} develops
robust welfare analysis for auction equilibria.
\Cref{s:robust-auction-design} studies robust revenue guarantees
without knowledge of the value distribution, with a spotlight on the sale of a
single item to one of two agents.
\Cref{s:robust-exchange} studies robust gains from trade in bilateral
exchange.
 
\section{Robust Analysis of Auction Equilibria}
\label{s:robust-analysis-auction-equilibria}

The first welfare theorem applies to markets without aggregate
uncertainty and with no unpriced externalities and says that
equilibrium prices support efficient outcomes \citep[see the
  textbook,][]{MWG-95}.  Auctions, on the other hand, often have both
aggregate uncertainty and unpriced externalities and both of these
lead to inefficiencies.

We present a simple analysis of robust welfare in
winner-pays-bid mechanisms, like the classic first-price auction, that
decomposes potential losses into ones of aggregate uncertainty and
ones of unpriced externalities.  For example, in Bayes--Nash equilibria
in auctions, there is generally aggregate uncertainty.  For example,
the winner-pays-bid highest-bids-win mechanisms in matching markets
do not have unpriced externalities while they do in combinatorial
auctions.  We will see that winner-determination algorithms, here
highest-bids-win, can impact the severity of unpriced externalities.
The main treatment will follow \citet{HHT-14,HHT-23}; however, many of
the ideas follow from common themes in the computer science literature on the
price of anarchy of auctions (further discussion is in
\Cref{s:robust-analysis-auction-equilibria-notes}).

As a starting motivation, notice that Bayes--Nash equilibria in the
first-price auction can be quite complex, even with asymmetric
uniform distributions (see \Cref{fig:asymmetric-first-price}).  The
robust approach can be used to quantify how bad equilibria can get in
various formats of auctions and understand what features enable good
performance.

\begin{figure}
  \centering
  \begin{subfigure}[t]{0.46\linewidth}
  \centering
  \begin{tikzpicture}[x=.80in,y=.80in]
    \draw[->] (0,0) -- (2.125,0) node[right] {$\val$};
    \draw[->] (0,0) -- (0,1.125) node[above] {$\bid$};
    \draw[dashed] (0,{2/3}) -- (2,{2/3});
    \draw[line width=1pt,domain=.02:1,samples=100,smooth]
      plot (\x,{(2/(3*\x))*(2-sqrt(4-3*\x*\x))});
    \draw[line width=1pt,domain=.02:2,samples=100,smooth]
      plot (\x,{(2/(3*\x))*(-2+sqrt(4+3*\x*\x))});
    \node at (.55,.80) {$\strati[1](\val)$};
    \node at (1.70,.43) {$\strati[2](\val)$};
    \foreach \x in {1,2} {
      \draw (\x,0) ++(0,-2.5pt) -- ++(0,5pt);
      \node[below=3pt] at (\x,0) {$\x$};
    }
    \draw (0,1) ++(-2.5pt,0) -- ++(5pt,0);
    \node[left=3pt] at (0,1) {$1$};
    \draw (0,{2/3}) ++(-2.5pt,0) -- ++(5pt,0);
    \node[left=3pt] at (0,{2/3}) {$\nicefrac23$};
    \node[below=3pt] at (0,0) {$0$};
  \end{tikzpicture}
  \caption{Equilibrium bids.}
  \label{fig:asymmetric-first-price-bids}
\end{subfigure}
\hfill
\begin{subfigure}[t]{0.46\linewidth}
  \centering
  \begin{tikzpicture}[x=.80in,y=.80in]
    \fill[gray!30,domain=.02:1,samples=100,smooth]
      plot ({1/sqrt(1/(\x*\x)-.75)},\x) -- (1,1) -- cycle;
    \draw[line width=1pt,domain=.02:1,samples=100,smooth]
      plot ({1/sqrt(1/(\x*\x)-.75)},\x);
    \draw[->] (0,0) -- (2.125,0) node[right] {$\vali[2]$};
    \draw[->] (0,0) -- (0,1.125) node[above] {$\vali[1]$};
    \draw[dashed] (0,0)--(1,1);
    \node at (.30,.70) {1 wins};
    \node at (1.35,.55) {2 wins};
    \foreach \x in {1,2} {
      \draw (\x,0) ++(0,-2.5pt) -- ++(0,5pt);
      \node[below=3pt] at (\x,0) {$\x$};
    }
    \draw (0,1) ++(-2.5pt,0) -- ++(5pt,0);
    \node[left=3pt] at (0,1) {$1$};
    \node[below=3pt] at (0,0) {$0$};
  \end{tikzpicture}
  \caption{Equilibrium allocation.}
  \label{fig:asymmetric-first-price-allocation}
\end{subfigure}
   \caption{First-price Bayes--Nash equilibrium with
    \(\vali[1] \sim \disti[1]=U[0,1]\) and \(\vali[2] \sim \disti[2]=U[0,2]\) from \citet{GLS-67}.
    Panel~\subref{fig:asymmetric-first-price-bids} gives the equilibrium bid
    strategies.  In panel~\subref{fig:asymmetric-first-price-allocation}, the
    gray region between the equilibrium boundary and the diagonal is
    allocated inefficiently.  Adapted from \citet{har-13}.}
  \label{fig:asymmetric-first-price}
\end{figure}

\subsection{Competitive Efficiency: Quantifying Inefficiencies from Unpriced Externalities}

In this section we will give a worst-case analysis of pure Nash
equilibrium welfare of the first-price auction and the winner-pays-bid
highest-bids-win combinatorial auction.  Focusing on deterministic
equilibria will showcase the competitive efficiency
definition as capturing the inefficiencies caused by unpriced
externalities.

Let \(\feasibles\) be the set of feasible allocations over $n$ agents and write
\(\epballocs(\bids)\) for the ex post allocation at bids \(\bids\).
For any nonnegative weight profile \(\zees\), the optimal feasible
objective is
\[
  \OPT(\zees)
  =\max_{\allocs\in\feasibles}\sum\nolimits_i \zeei\,\alloci.
  \]
  In a first-price auction, $\feasibles = \{\allocs \in [0,1]^n : \sum_i \alloci \leq 1 \}$ and $\OPT(\zees) = \max_i \zeei$.
  
Let \(\SB_{\star}\) denote the welfare-optimal Vickrey--Clarke--Groves
mechanism \citep{vic-61,cla-71, gro-73}.  It does not depend
on \(\dist\) and implements the pointwise welfare optimum, so second-best and first-best coincide:
\[
  \SB_{\star}(\vals)=\FB(\vals)=\OPT(\vals).
\]
  Three quantities that will be important in the subsequent analysis are
a bidder \(i\)'s critical bid, the realized revenue, and realized welfare:
\begin{align*}
  \critbidi(\bidsmi)
  &=\inf\{\bidi' : \epballoci(\bidi',\bidsmi)=1\},\\
  \Rev{\bids}
  &=\sum\nolimits_i\bidi\epballoci(\bids),\\
  \Wel(\vals,\bids)
  &=\sum\nolimits_i\vali\epballoci(\bids).
\end{align*}

The competition that a bidder $i$ faces, given the bids of others
$\bidsmi$ and the bid-allocation rule of the auction $\epballocs$, can
be summarized by their critical bid $\critbidi(\bidsmi)$.

In a first-price auction the critical bid is the highest other bid.
Consider the first-price auction with value profile $\vals =
(10,20,30)$, which we will use below as a running example.  If the bid
profile is $\bids = (5,10,15)$, the critical bid profile is $\critbids
= (15,15,10)$.

Competitive efficiency (\Cref{def:competitive-efficiency}) describes
the extent to which the aggregate competition drives revenue.  It is a
property of the mechanism, rather than of an equilibrium or a value
distribution.

\begin{definition}[Competitive efficiency]
\label{def:competitive-efficiency}
A winner-pays-bid mechanism is \(\CompEff \in [0,1]\) competitively efficient if,
for every bid profile \(\bids\),
\[
  \Rev{\bids}\geq\CompEff\,\OPT(\critbids).
\]
\end{definition}

In the example, the revenue on $\bids =
(5,10,15)$ is $15$ and the aggregate competition is $\OPT(\critbids) =
15$; thus, $\CompEff = 1$.  In fact, $\CompEff = 1$ in general for the
first-price auction.

\begin{proposition}[\citealp{HHT-23}]
\label{prop:single-item-competitive-efficiency}
The first-price auction is \(1\)-competitively efficient.
\end{proposition}

\begin{proof}
 The argument is the following sequence:
  \begin{alignat*}{3}
\OPT(\critbids)
&= \max\nolimits_{\allocs \in \feasibles}
   \sum\nolimits_i \critbidi\alloci
&&= \max\nolimits_{\allocs \in \feasibles}
   \sum\nolimits_i
   [\max\nolimits_{j\neq i}\bidi[j]]\,\alloci \\
&\leq \max\nolimits_{\allocs \in \feasibles}
   \sum\nolimits_i
   [\max\nolimits_j\bidi[j]]\,\alloci
&&= [\max\nolimits_j\bidi[j]]
   \max\nolimits_{\allocs \in \feasibles}
   \sum\nolimits_i\alloci
 &= \Rev{\bids} \cdot 1.
\end{alignat*}
Start with the definition of optimal welfare according to critical
bids.  The second step plugs in the definition of critical bids for
the first-price auction. The inequality replaces this bidder-specific
maximum with the highest bid overall.  The fourth step factors this
common highest bid out of the sum.  Conclude by observing that a
feasible allocation assigns at most one item, and the highest bid is
the first-price auction's revenue.
\end{proof}

\Cref{prop:single-item-competitive-efficiency}
reflects the absence of unpriced externalities: serving one bidder
excludes at most one other bidder, and that competition appears
one-for-one in the payment.

In a pure Nash equilibrium there is no aggregate uncertainty.  Bidders
bid with full knowledge of the outcome they will receive.  Individual
efficiency quantifies the fraction of a bidder's value that is
accounted for by the sum of the bidder's utility and their critical bid.
If their critical bid is high, then their utility can be low.  With no
aggregate uncertainty, winner-pays-bid mechanisms are individually
efficient.

\begin{definition}[Individual efficiency with deterministic bids]
\label{def:deterministic-individual-efficiency}
A deterministic bidding outcome is \(\IndEff\)-individually efficient
if every bidder satisfies
\[
  \utili+\critbidi\geq\IndEff\,\vali.
\]
\end{definition}

Returning to the example of $\vals = (10,20,30)$ in the first-price
auction, bids $\bids =(10,20,20)$ are a pure Nash equilibrium (if ties
are broken for bidder 3).  Critical bids are $\critbids = (20,20,20)$
and the utilities are $\utils = (0,0,10)$.  Bidder 3 satisfies
$\utili[3] + \critbidi[3] = 10 + 20 = 30 = \vali[3]$.  Bidders 1 and
2 also satisfy their respective inequalities; thus, $\IndEff = 1$.

Pure Nash equilibria in first-price auctions need not exist for arbitrary
deterministic tie-breaking, so
\Cref{prop:deterministic-individual-efficiency} is stated for
approximate pure Nash equilibria.  A pure \((1-\epsilon)\)-Nash
equilibrium is a bid profile at which each bidder's utility is at
least a \(1-\epsilon\) fraction of the utility from her best response.

\begin{proposition}[\citealp{HHT-23}]
\label{prop:deterministic-individual-efficiency}
A pure \((1-\epsilon)\)-Nash equilibrium of a winner-pays-bid mechanism
is \((1-\epsilon)\)-individually efficient.
\end{proposition}

\begin{proof}
To avoid algebra, consider an exact best response with ties favoring the
bidder.

Since we are considering pure strategies, bidder $i$ faces a
deterministic threshold at $\critbidi$.  Bids $\bidi \geq \critbidi$
win; bids $\bidi < \critbidi$ lose.  One optimal bid is \(\bidi =
\min(\vali,\critbidi)\).  If \(\vali<\critbidi\), the bidder loses,
and \(\utili=0\); so \(\utili+\critbidi = \critbidi>\vali\).  If
\(\vali\geq\critbidi\), the bidder wins and
\(\utili+\critbidi=(\vali-\critbidi)+\critbidi=\vali\).

The same cases, with approximate-best-response slack and some more algebra, give the result for a pure \((1-\epsilon)\)-Nash equilibrium
under any deterministic tie-breaking rule.
\end{proof}

Competitive and individual efficiency multiply to give robust welfare guarantees.

\begin{theorem}[\citealp{HHT-23}]
\label{thm:deterministic-welfare-guarantee}
If a mechanism is \(\CompEff\)-competitively efficient
and a deterministic bidding outcome is \(\IndEff\)-individually
efficient, then the outcome obtains at least a
\(\CompEff\IndEff\) fraction of optimal welfare.
\end{theorem}

\begin{proof}
Fix \(\vals\) and an allocation \(\allocs\opt\) attaining
\(\OPT(\vals)\).  Individual efficiency for agent $i$ gives
\[
  \utili+\critbidi\geq\IndEff\,\vali.
\]
Multiplying by \(\alloci\opt\in[0,1]\) and using \(\utili\geq 0\) gives
\[
  \utili+\critbidi\alloci\opt
  \geq\IndEff\,\vali\alloci\opt.
\]
Summing over bidders gives
\[
  \sum\nolimits_i\utili
  +\sum\nolimits_i\critbidi\alloci\opt
  \geq
  \IndEff\sum\nolimits_i\vali\alloci\opt.
\]
Because \(\allocs\opt\) is feasible, its critical-bid objective $\sum\nolimits_i\critbidi\alloci\opt$ is at
most \(\OPT(\critbids)\).  Competitive efficiency then bounds this
term by revenue, giving
\[
  \sum\nolimits_i\utili
  +\nicefrac{1}{\CompEff}\Rev{\bids}
  \geq\IndEff\,\OPT(\vals).
\]
Since total utility plus revenue is welfare and
\(\CompEff\leq1\), increasing the left-hand side gives
\[
  \nicefrac{1}{\CompEff}\Wel(\vals,\bids)
  \geq\IndEff\,\OPT(\vals). \qedhere
\]
\end{proof}

\begin{corollary}
\label{cor:first-price-deterministic-welfare}
Every pure \((1-\epsilon)\)-Nash equilibrium of the first-price auction
obtains at least a \(1-\epsilon\) fraction of optimal welfare.
\end{corollary}

For the first-price auction, competitive efficiency contributes no
loss.  With deterministic bids, individual efficiency contributes only
(possibly) the loss from (possibly) approximate best response.
Indeed, \Cref{cor:first-price-deterministic-welfare} is also a
consequence of the first welfare theorem.

The preceding analysis separates losses due to the response problem
(individual efficiency) from losses due to the allocation rule
(competitive efficiency).  We will now illustrate the point that
competitive efficiency captures losses from unpriced externalities
using the example of single-minded combinatorial auctions
\citep[cf.][]{LOS-02}.

\begin{definition}
  \label{def:single-minded-combinatorial-auction}
In a single-minded combinatorial auction, there are \(m\) items, and
bidder \(i\) has value \(\vali\) for receiving every item in a known
bundle \(\demandseti\) of items $\{1, \ldots, m\}$.  A feasible
allocation serves bidders with nonoverlapping bundles, i.e.,
$\feasibles = \{ \allocs \in \{0,1\}^n : \alloci = \alloci[j] = 1
\Rightarrow \demandseti \cap \demandseti[j] = \emptyset \text{ for all
} i\ne j\}$.
\end{definition}

The natural generalization of the first-price auction to single-minded
combinatorial auctions is to solicit bids, select as winners the
feasible set of bidders with highest total bid, and charge all winners
their bids.  We call this the winner-pays-bid highest-bids-win auction.

Consider \(m+2\) bidders.  For each \(i\in\{1,\ldots,m\}\), bidder
\(i\) values the singleton bundle \(\demandseti=\{i\}\) at one.  Each
of bidders \(m+1\) and \(m+2\) values the grand bundle
\(\{1,\ldots,m\}\) at one.  Suppose that the two grand-bundle bidders
bid one and the singleton bidders bid zero.  One of the grand-bundle
bidders wins.  Every bidder has utility zero and no bidder can obtain
positive utility by deviating, so these bids form a pure Nash
equilibrium.  The welfare of this Nash equilibrium is one, while the
optimal welfare is $m$ by assigning each item $i$ to bidder $i$.  The
ratio of equilibrium welfare to optimal welfare is $\nicefrac1m$.

The competitive efficiency matches this inefficiency.  The revenue is
$\Rev{\bids} = 1$.  On the other hand, every singleton bidder has a
critical bid of one; thus, $\OPT(\critbids)=m$ and the competitive
efficiency is $\CompEff = \nicefrac1m$.  Omitting the details, a
matching lower bound on the competitive efficiency holds for every
profile of demanded bundles and bids; thus, the mechanism's
competitive efficiency is $\nicefrac1m$.  By
\Cref{thm:deterministic-welfare-guarantee}, $\nicefrac1m$ is the
fraction of optimal welfare guaranteed in any pure Nash equilibrium of
the single-minded combinatorial auction.

In this example there are unpriced externalities.  The two
grand-bundle bidders competed with each other on price, but this price
does not reflect the externality from not serving any of the
single-item bidders.  Competitive efficiency quantifies exactly
the loss due to this unpriced externality.

\subsection{Individual Efficiency: Quantifying Inefficiency from Aggregate Uncertainty}

We will now consider environments with aggregate uncertainty due to both
mixed equilibria and incomplete information.  The first step is to
lift competitive-efficiency bounds from deterministic to randomized
environments, i.e., randomized bids.  This lifting follows from a
closure property: competitive efficiency is closed under randomization.
In other words, the worst-case
competitive efficiency is deterministic.  Thus, to analyze the
randomized competitive efficiency of a mechanism, it suffices to
analyze the competitive efficiency on deterministic bids.

Define a bidder $i$'s expected critical bid as
\(\cbidi=\expect{\smash{\critbidi}}\).

\begin{definition}[Competitive efficiency with random bids]
A winner-pays-bid mechanism is \(\CompEff\)-competitively efficient
with random bids if, for every distribution over bid profiles,
\[
  \expect[\bids]{\Rev{\bids}}
  \geq\CompEff\,\OPT(\cbids).
\]
\end{definition}

\begin{proposition}[\citealp{HHT-23}]
\label{prop:competitive-efficiency-closure}
A mechanism is \(\CompEff\)-competitively efficient for deterministic
bids if and only if it is \(\CompEff\)-competitively efficient for random
bids.
\end{proposition}

\begin{proof}
The deterministic-to-randomized direction is the following sequence:
\begin{align*}
\expect{\Rev{\bids}}
&\geq \expect{\CompEff\,\OPT(\critbids)}
= \expect{\CompEff\,
   \max\nolimits_{\allocs\in\feasibles}
   \sum\nolimits_i\alloci\critbidi}
\geq \CompEff
   \max\nolimits_{\allocs\in\feasibles}
   \sum\nolimits_i\alloci\expect{\smash{\critbidi}}
= \CompEff\,\OPT(\cbids).
\end{align*}
The first step applies deterministic competitive
efficiency to every realized bid profile and then takes expectations.
The second step expands the optimal critical-bid
objective.  The third step uses the fact that the
expectation of a maximum is at least the maximum of the expectations.
The final step uses the definition of expected
critical bids and $\OPT(\cdot)$.

The reverse direction follows because a deterministic
bid profile is a degenerate distribution over bid profiles.
\end{proof}

More generally, \citet{HHT-23} show that competitive efficiency is
also closed under randomization of the mechanism, agent selection, and
correlated signaling.
Essentially, deterministic full information is the worst case.

When bids are randomized, critical bids are randomized, and bidders
face aggregate uncertainty in their best-response problem.  As a
bidder's bid varies, their probability of winning traces an
interim bid-allocation rule.  The geometry of this best-response
problem is depicted in \Cref{fig:individual-efficiency-geometry},
where the utility from a bid is the dark rectangle below the bid-allocation rule,
while the expected critical bid is the light area above it.  We will
extend the definition of individual efficiency to randomized bids by
requiring that it hold in expectation, i.e., by replacing $\critbidi$
with $\cbidi$.

\begin{figure}
  \centering
  \begin{subfigure}[t]{0.47\textwidth}
    \centering
    \resizebox{\linewidth}{!}{\begin{tikzpicture}[x=4.4cm,y=2.13cm]
  \fill[gray!22]
    (0,1) -- plot[smooth] coordinates
      {(0,0) (.08,.30) (.20,.46) (.50,.55) (.80,.62) (1.02,.82) (1.08,1)}
    -- cycle;
  \fill[gray!58] (.20,0) rectangle (1,.46);
  \draw[line width=.9pt]
    plot[smooth] coordinates
      {(0,0) (.08,.30) (.20,.46) (.50,.55) (.80,.62) (1.02,.82) (1.08,1)};
  \draw[dashed,thick] (0,0) rectangle (1,1);
  \draw[->] (0,1.05) -- (0,0) -- (1.16,0);
  \node[below] at (1,0) {$\val$};
  \node[below] at (.20,0) {$\bid$};
  \node[left] at (0,.46) {$\balloc(\bid)$};
  \node[left] at (0,1) {$1$};
  \node[anchor=west] at (1.02,.70) {$\balloc(\cdot)$};
  \node at (.60,.23) {$\butil(\val,\bid)$};
  \node at (.39,.78) {$\cbid=\expect{\smash{\critbid}}$};
\end{tikzpicture}
 }
    \caption{Best-response geometry.}
    \label{fig:individual-efficiency-geometry}
  \end{subfigure}
  \hfill
  \begin{subfigure}[t]{0.47\textwidth}
    \centering
    \resizebox{\linewidth}{!}{\begin{tikzpicture}[x=4.4cm,y=2.13cm]
  \coordinate (M) at (.50,.55);
  \fill[gray!22] (0,.55) rectangle (.50,1);
  \fill[gray!58] (.50,0) rectangle (1,.55);
  \draw[line width=.9pt]
    plot[smooth] coordinates
      {(0,0) (.08,.30) (.20,.46) (.50,.55) (.80,.62) (1.02,.82) (1.08,1)};
  \draw[dashed,thick] (0,0) rectangle (1,1);
  \draw[dashed] (.50,0) -- (.50,.55) -- (0,.55);
  \draw[->] (0,1.05) -- (0,0) -- (1.16,0);
  \fill (M) circle (.015);
  \node[below] at (.50,0) {$\nicefrac{\val}{2}$};
  \node[below] at (1,0) {$\val$};
  \node[left] at (0,.55) {$\balloc(\nicefrac{\val}{2})$};
  \node[left] at (0,1) {$1$};
  \node at (.72,.83) {$\balloc(\cdot)$};
  \node at (.75,.27) {$\butil(\val,\nicefrac{\val}{2})$};
  \node at (.25,.78) {part of $\cbid$};
\end{tikzpicture}
 }
    \caption{The half-value deviation.}
    \label{fig:individual-efficiency-half}
  \end{subfigure}
  \caption{Individual-efficiency geometry in a winner-pays-bid mechanism.
    Panel~\subref{fig:individual-efficiency-geometry} depicts the utility from
    bid \(\bid\) as the dark rectangle, the
    expected critical bid as the light region, and value \(\val\) as the area
    of the dashed rectangle.  Panel~\subref{fig:individual-efficiency-half}
    gives the
    \(\nicefrac12\)-individual-efficiency guarantee from bidding \(\nicefrac{\val}{2}\): the dark
    utility rectangle and the light rectangle contained in the
    expected-critical-bid region have total area \(\nicefrac{\val}{2}\).  Adapted from
    \citet{har-13}.}
  \label{fig:individual-efficiency}
\end{figure}

\begin{definition}[Individual efficiency with random bids]
\label{def:individual-efficiency}
A bidding outcome is \(\IndEff\)-individually efficient if every bidder $i$
satisfies
\[
  \utili+\cbidi\geq\IndEff\,\vali.
\]
\end{definition}

\begin{proposition}[\citealp{ST-13,HHT-23}]
\label{prop:first-price-individual-efficiency}
Every best response in a winner-pays-bid bidding problem is
\((1-\nicefrac1e)\)-individually efficient.
\end{proposition}

To prove \Cref{prop:first-price-individual-efficiency}, we need to
show that the white area in the dashed rectangle is small relative to
the rectangle's area (which is $\val$; see
\Cref{fig:individual-efficiency-geometry}).  This proof can be
obtained by solving a zero-sum game, like the ones in
\Cref{s:robustness-paradigms}, whose value is $\IndEff =
1-\nicefrac1e$.  A simpler $\IndEff \geq \nicefrac12$ bound can be
shown geometrically by lower-bounding the utility by the utility from
the deviation $\bid = \nicefrac{\val}{2}$ (see
\Cref{fig:individual-efficiency-half}).

The deterministic welfare proof, that combines competitive efficiency
$\CompEff$ with individual efficiency $\IndEff$ to give robust welfare
bounds, extends to random bids by replacing each critical bid with its
expectation and taking expectations when summing over bidders (details
omitted).  When values are random as in Bayes--Nash equilibrium, the preceding quantities \(\utili\)
and \(\cbidi\) are interpreted as interim quantities conditional on
bidder \(i\)'s value.

\begin{theorem}[\citealp{HHT-23}]
\label{thm:welfare-guarantee}
If a mechanism is \(\CompEff\)-competitively efficient and every best
response is \(\IndEff\)-individually efficient, then every Bayes--Nash
equilibrium obtains at least a \(\CompEff\IndEff\) fraction of expected
optimal welfare.
\end{theorem}

\begin{corollary}
\label{cor:first-price-bayes-welfare}
Every Bayes--Nash equilibrium of the first-price auction obtains at
least a \(1-\nicefrac1e\) fraction of expected optimal welfare.
\end{corollary}

Note that the same individual efficiency guarantee applies to any
winner-pays-bid mechanism.  For example, it holds equally for the
first-price auction and the winner-pays-bid highest-bids-win mechanism
for single-minded combinatorial auctions.  To show a new
winner-pays-bid mechanism is robustly efficient, it suffices only to
analyze its (deterministic) competitive efficiency.  The fraction of
welfare lost from aggregate uncertainty in any winner-pays-bid
mechanism is at most $\nicefrac1e$.

\subsection{Competitive Efficiency of Greedy Mechanisms: Reducing Unpriced Externalities}

Return to the single-minded combinatorial auction.  The
highest-bids-win rule has competitive efficiency \(\nicefrac1m\), but
there is a better winner-determination rule that reduces the loss from
unpriced externalities.  A greedy rule orders bidders by a priority
function of their bids and allocates to each bidder in turn whenever
their bundle remains feasible.  The corresponding winner-pays-bid
mechanism charges each winner their bid.

For \(\apx\in[0,1]\), an allocation rule is an \(\apx\)-approximation if
it obtains at least an \(\apx\) fraction of the maximum feasible sum of
bids for every bid profile.
For single-item multi-unit feasibility, greedy by bid is optimal and the associated
winner-pays-bid mechanism is \(1\)-competitively efficient.  For
single-minded combinatorial auctions, the following classical
approximation gives the relevant rule.

\begin{theorem}[\citealp{LOS-02}]
\label{thm:greedy-single-minded-approximation}
For a single-minded combinatorial auction with \(m\) items, the greedy
priority rule \(\bidi/\sqrt{|\demandseti|}\) obtains an
\(\apx=\nicefrac1{\sqrt m}\) fraction of the maximum feasible sum of bids.
\end{theorem}

This guarantee transfers to competitive efficiency
through coalitional non-bossiness.

\begin{definition}[Coalitional non-bossiness]
An allocation rule $\epballocs$ is coalitionally-non-bossy if any coalition of
losers can change their bids while remaining below their critical bids
without changing the allocation.
\end{definition}

Notice that highest-bids-win for the single-minded combinatorial
auction is not coalitionally-non-bossy.  The singleton bidders in the
previous example could jointly increase their bids to $\nicefrac1m+\epsilon$
each, for sufficiently small $\epsilon>0$, after which they are
allocated instead of one of the grand-bundle bidders.

\begin{lemma}[\citealp{HHT-23}]
\label{lem:greedy-coalitional-nonbossiness}
Every greedy-by-priority allocation rule is coalitionally-non-bossy.
\end{lemma}

\begin{proof}
  The greedy algorithm only makes decisions based on feasibility
  relative to the set of selected bidders.  Consider changing any
  loser's bid while keeping it below their critical bid.  With this
  change, the loser never enters the winning set; thus, the same
  decisions are made for all other agents.  By induction, the same
  conclusion holds for any set of losers.
\end{proof}

\begin{proposition}[\citealp{HHT-23}]
\label{prop:greedy-combinatorial-competitive-efficiency}
A coalitionally-non-bossy winner-pays-bid mechanism whose allocation
rule is an \(\apx\)-approximation is
\(\CompEff = \apx\) competitively efficient.
\end{proposition}

\begin{proof}
Fix bids \(\bids\) and their critical bids \(\critbids\).  Let
\(\bids'\) leave the bids of winners unchanged and raise the bids of
losers to their critical bids while retaining the losing outcome.  Observe,
\begin{align*}
\Rev{\bids}
= \sum\nolimits_i\bidi\epballoci(\bids)
= \sum\nolimits_i\bidi'\epballoci(\bids')
\geq {\apx}\OPT(\bids')
\geq {\apx}\OPT(\critbids).
\end{align*}
Start with the definition of revenue.  The second step holds because only losing bids
change and coalitional non-bossiness preserves the allocation.  The
third step applies the allocation rule's approximation guarantee to
\(\bids'\).  For winners, the unchanged bid is at least the critical
bid; for losers, the modified bid is the critical bid so $\bids' \geq
\critbids$ coordinatewise.  Monotonicity of the $\OPT(\cdot)$ gives
the final step.
\end{proof}

Combining
\Cref{thm:greedy-single-minded-approximation,prop:greedy-combinatorial-competitive-efficiency}
shows that greedy by bid per square root of bundle size has competitive
efficiency \(\nicefrac1{\sqrt m}\).  Combining this guarantee with the
individual-efficiency guarantee \(1-\nicefrac1e\) gives a robust-efficiency guarantee of
\((1-\nicefrac1e)\,\nicefrac1{\sqrt m}\).  The square-root loss is due to the many-for-one
unpriced externalities in the combinatorial allocation problem; the
factor \(1-\nicefrac1e\) is the separate loss from aggregate uncertainty.  The
greedy rule reduces the loss from unpriced externalities
relative to the highest-bids-win rule.

\subsection{Notes}
\label{s:robust-analysis-auction-equilibria-notes}

Computer science conventionally defines the price of anarchy as the
worst-case ratio of optimal welfare to equilibrium welfare, the
reciprocal of the performance ratio used in this survey.  The price of
anarchy was introduced by \citet{KP-99} and applied to selfish routing
by \citet{RT-02}.  There is a long literature studying the price of
anarchy in many different families of games.  \citet{rou-09} observes
that many price-of-anarchy analyses follow a common template and
formalizes this template through a parameterized definition of
smoothness.

Early price-of-anarchy analyses of auctions include
combinatorial auctions \citep{CKS-08, BL-10} and simultaneous
first-price auctions \citep{CKST-13}.  \citet{rou-12} and
\citet{ST-13} extended the smoothness framework to games of incomplete
information and to auctions, respectively.  For auctions, a smoothness
proof combines a deviation for each bidder with an accounting argument
comparing utilities plus revenue with optimal welfare.

\citet{HHT-14} separates the inefficiencies quantified in the
smoothness framework into two parts.  Individual efficiency quantifies
how efficiently a bidder can solve her best-response problem, while
competitive efficiency quantifies how the mechanism converts
competition into revenue.  The former measures aggregate uncertainty;
the latter measures the unpriced externalities in the mechanism.  The
terminology and economic interpretation presented here follow the expanded treatment of \citet{HHT-23}.
They also study the all-pay payment format, the revenue
objective, and many other environments.

The analysis of the $\IndEff = 1-\nicefrac1e$ individual efficiency follows
from a randomized-deviation proof of \citet{ST-13}; \citet{HHT-14}
give an alternative geometric proof.  The smooth-mechanism analysis of
\citet{ST-13} also allows multi-dimensional preferences and
simultaneous and sequential composition.  The connection between
competitive efficiency and greedy approximation recasts the
greedy-auction analysis of \citet{BL-10} in the
individual-efficiency-competitive-efficiency framework.

For single-minded combinatorial auctions, \citet{DK-15} show that the
competitive efficiency of $\CompEff = \nicefrac1{\sqrt{m}}$ is optimal
up to a constant factor for winner-pays-bid mechanisms.  In
particular, the best competitive efficiency achievable by any
winner-pays-bid mechanism for the single-minded combinatorial auction
environment is $O(\nicefrac1{\sqrt m})$.
 
\section{Robust Auction Design}
\label{s:robust-auction-design}

Welfare and revenue are distinct objectives.  For welfare
maximization, there is a single mechanism that maximizes welfare
pointwise on all profiles of values, i.e., the so-called
Vickrey--Clarke--Groves mechanism \citep{vic-61,cla-71, gro-73}.  In
contrast, for revenue maximization, the revenue-optimal mechanism in
the classical Bayesian setting of \citet{har-67} and \citet{mye-81}
depends on the prior distribution from which the agents' preferences
are drawn.  In fact, welfare is unique among objectives in that it is
optimized pointwise.  In most settings of mechanism design, the
optimal mechanism depends on the distribution of agent preferences.

Naturally, the prior-dependent optimal mechanism scales with the
distribution of preferences.  For example, the revenue-optimal
mechanism for selling a single item to agents with i.i.d.\ and uniformly
distributed values on $[0,1]$ is to run a second-price auction with a
reserve of $\nicefrac12$.  If the value distribution is scaled up to be uniform
on $[0,100]$, the optimal mechanism scales as well to use a reserve
price of $\nicefrac12 \times 100 = 50$.

The prior-robust auction design problem aims to develop a VCG-like
result but for revenue-like objectives.  A single mechanism is sought
that works well on a large family of distributions.  Most analyses
consider a family of distributions closed under scaling.\footnote{Note that the example of \Cref{s:robustness-paradigms} was
not closed under scaling.  This was needed to ensure that the other
robustness paradigms had non-trivial answers.}

The first papers in this area showed that it is possible to develop
nearly optimal auctions with large supply and large demand in the
limit; essentially the mechanism can do market analysis on the fly
\citep{GHW-01, seg-03, BV-03}.  When the optimal mechanism sells to
many agents, the market-analysis-on-the-fly approach is broadly
successful.  The recent literature has shifted to ask what is possible
with small supply and small markets; the canonical setting is one
with a single item and agents with values drawn i.i.d.\ from an unknown
but regular distribution (see \Cref{def:regular-mhr}, below).

We will see that the theorem of \citet{BK-96}, which shows that
recruiting one more agent improves revenue more than setting the
optimal reserve, implies a robust ratio guarantee.  Specifically, with
$n$ i.i.d.\ regularly distributed agents and a single item, the
second-price auction obtains at least a $1-\nicefrac1n$ fraction of the
optimal revenue.  Specifically, for $n=2$ the auction obtains at least half the
optimal revenue.  This $n=2$ result can also be seen by an elegant
geometric proof due to \citet{DRY-15}.  When the supply $k$ is
increased from $1$ to $n$, the general bound degrades as $1-\nicefrac{k}{n}$, which
is trivial for $k=n$.  Several simple mechanisms recover the $\nicefrac12$
guarantee for general $k$ and $n$, for example, limiting supply to
approximately half the number of bidders and running a welfare-maximizing
mechanism \citep{RTY-12}.
The final topic of this section is the question of whether the bound
of $\nicefrac12$ for the second-price auction with $k=1$ item and $n=2$ agents
is optimal or can be improved.  In fact, it can be improved.

\subsection{The Prior-robust Framework}

The prior-robust auction design framework chooses an auction before the value
distribution is known and evaluates its expected revenue against the Bayesian
optimal auction for that distribution.  This section develops the
single-unit and $k$-unit guarantees obtained from competition, gives two
$k$-unit mechanisms whose guarantees do not deteriorate with supply, and
summarizes the current optimal-design results for two bidders.

Consider $n$ unit-demand bidders whose values are drawn independently
and identically from a common distribution $\dist$.  Write
$\mech(\dist)$ for the expected revenue of mechanism $\mech$ and
$\SB(\dist)=\SB_{\dist}(\dist)$ for the expected revenue of the
Bayesian-optimal auction for $\dist$.  Throughout this section,
$\Mechs$ denotes the family of Bayesian-incentive-compatible
interim-individually-rational mechanisms for the stated environment.

This section restricts attention to i.i.d.\ distributions that satisfy
a natural convexity property known as regularity \citep{mye-81}.  Some
results impose the stronger monotone hazard rate property, which is
satisfied by many common distributions, including uniform, normal, and
exponential distributions.

\begin{definition}[Regular and monotone-hazard-rate distributions]
\label{def:regular-mhr}
For quantile $\quant\in[0,1]$, let $\Val(\quant)=\dist^{-1}(1-\quant)$
denote inverse demand and let $\rev(\quant)=\quant\,\Val(\quant)$
denote the revenue curve (\Cref{fig:single-sample-inverse-demand}).  A
distribution is regular if its revenue curve is concave
(\Cref{fig:single-sample-revenue-curve}).  If $\dist$ has density
$\dens$, it satisfies the monotone hazard rate property if its hazard
rate $\dens(\val)/(1-\dist(\val))$ is nondecreasing in $\val$.
Monotone hazard rate implies regularity.
\end{definition}

\begin{definition}[Prior-robust max-min ratio]
\label{def:prior-independent-max-min-ratio}
Fix an environment, a mechanism family $\Mechs$, and a distribution family
$\Dists$.  The prior-robust max-min ratio is
\[
  \OptPIratio
  =\sup_{\mech\in\Mechs}\inf_{\dist\in\Dists}
    \frac{\mech(\dist)}{\SB(\dist)}.
\]
The inner infimum is a fixed mechanism's prior-robust ratio; a ratio with $\SB(\dist)=0$ is one.
\end{definition}

\subsection{Single-unit Prior-robust Auctions}

The theorem of \citeauthor{BK-96} is central to the development of
prior-robust mechanisms.  It assumes that bidders' values are drawn
i.i.d.\ from a regular distribution (\Cref{def:regular-mhr}).  It shows
that a
welfare-maximizing single-item auction, e.g., the second-price auction, with
one additional bidder obtains at least as much revenue as the
revenue-optimal auction (without the additional bidder).  Its proof, below, is due to \citet{kir-06}.

\begin{theorem}[\citealp{BK-96,kir-06}]
\label{thm:bulow-klemperer-resource-augmentation}
For $n$ i.i.d.\ regular bidders in a single-unit environment, the
welfare-maximizing auction with $n+1$ bidders has expected revenue at least the
Bayesian optimal revenue with the original $n$ bidders.
\end{theorem}

\begin{proof}
Let $\RM(n)$ denote the expected revenue of the revenue-maximizing auction
with $n$ bidders, and let $\WM(n+1)$ denote the expected revenue of the
welfare-maximizing auction with $n+1$ bidders.

First, consider the revenue-maximizing auction among auctions for
$n+1$ bidders that always sell the item.  Expected revenue equals
expected virtual surplus \citep{mye-81}.  Because the bidders are
identically distributed and the distribution is regular, virtual value
is nondecreasing in value \citep{mye-81}.  Thus, subject to always
selling, expected virtual surplus is maximized by allocating the item
to a highest-valued bidder.  The second-price auction implements this
rule, so $\WM(n+1)$ is optimal among auctions that
always sell.

Now define an auction $B$ for $n+1$ bidders.  Auction $B$ simulates the
revenue-maximizing auction on bidders $1,\ldots,n$.  If that auction does not
sell the item, then $B$ gives the item to bidder $n+1$ for free.  The
allocation and payment of bidder $n+1$ do not depend on her report, while those
of the first $n$ bidders come from the simulation.  Thus, $B$
is incentive compatible, always sells the item, and has expected revenue $\RM(n)$.

Optimality of the second-price auction among auctions that always sell
gives $\WM(n+1)\geq B(n+1)=\RM(n)$, as claimed.
\end{proof}

Although \Cref{thm:bulow-klemperer-resource-augmentation} is stated to
compare markets of different sizes, it can also be used to compare the
welfare-maximizing and revenue-optimal mechanisms in the same market.  This
comparison gives a prior-robust guarantee.

\begin{corollary}[\citealp{DRY-15}]
\label{cor:second-price-worst-case-ratio}
For $n\geq2$ i.i.d.\ regular bidders and the revenue objective, the
$n$-bidder second-price auction has prior-robust ratio at least
$1-\nicefrac{1}{n}$.
\end{corollary}

\begin{proof}
Fix a regular distribution $\dist$, and write $\RM(n)$ and $\WM(n)$ for the
expected revenues of the revenue-maximizing and welfare-maximizing auctions,
respectively, with $n$ bidders.

Construct an auction $B$ for $n-1$ bidders as follows.  Draw an independent
dummy value from $\dist$ and run the revenue-maximizing $n$-bidder auction on
the $n-1$ bidders together with the dummy bidder.  If the dummy bidder wins,
reject everyone; otherwise, implement the outcome of the simulated auction
for the real bidders.

The simulated bidders are exchangeable, so each contributes one $n$th of
the expected revenue of the $n$-bidder revenue-maximizing auction.  Auction
$B$ retains the payments of the $n-1$ real bidders and discards the payment
of the dummy bidder.  Consequently, $B(n-1)=(1-\nicefrac1n)\RM(n)$.  Revenue
optimality with $n-1$ bidders therefore implies
$\RM(n-1)\geq(1-\nicefrac1n)\RM(n)$.

Applying \cref{thm:bulow-klemperer-resource-augmentation} with $n-1$
original bidders gives $\WM(n)\geq\RM(n-1)$.  Combining the two inequalities
gives $\WM(n)\geq(1-\nicefrac1n)\RM(n)$, as claimed.
\end{proof}

The sample-based results below are intermediate steps toward
prior-robust mechanisms that do not have sample access.

Another interpretation of the theorem of \citeauthor{BK-96}, from the
$n=1$ special case, is that a random price (from the distribution)
guarantees half the revenue of the optimal, a.k.a., monopoly, price.
Consider the monopoly-pricing problem for a buyer whose value is drawn
from $\dist$.  Monopoly pricing with a random reserve draws
$\critval\sim\dist$ independently of the buyer's value and offers the
item at price $\critval$.  \citet{DRY-15} give an elegant geometric 
analysis, based on the \citet{BR-89} reinterpretation of virtual values
as marginal revenue,  proving this special case of the theorem. 

\begin{figure}
  \centering
  \begingroup
\begin{subfigure}[t]{0.31\linewidth}
\centering
\begin{tikzpicture}[
  x=2.65cm,
  y=1.43cm,
  >=Stealth,
  every node/.style={font=\scriptsize}
]
  \fill[gray!22] (0,0) rectangle (.72,.28);
  \node at (.36,.14) {$\rev(\exquant)$};
  \draw[->] (0,0) -- (1.08,0) node[below right] {$\quant$};
  \draw[->] (0,0) -- (0,1.08) node[above] {$\Val(\quant)$};
  \draw (0,.025) -- (0,-.025) node[below] {$0$};
  \draw (1,.025) -- (1,-.025) node[below] {$1$};
  \draw (-.025,1) node[left] {$1$} -- (.025,1);
  \draw[very thick] (0,1) -- (1,0);
  \draw[dashed] (.72,0) -- (.72,.28) -- (0,.28);
  \fill (.72,.28) circle (1.4pt);
  \node[below] at (.72,0) {$\exquant$};
  \node[left] at (0,.28) {$\Val(\exquant)$};
\end{tikzpicture}

\caption{Inverse demand.}
\label{fig:single-sample-inverse-demand}
\end{subfigure}
\hfill
\begin{subfigure}[t]{0.31\linewidth}
\centering
\begin{tikzpicture}[
  x=2.65cm,
  y=1.43cm,
  >=Stealth,
  every node/.style={font=\scriptsize}
]
  \draw[->] (0,0) -- (1.08,0) node[below right] {$\quant$};
  \draw[->] (0,0) -- (0,1.08) node[above] {$\rev(\quant)$};
  \draw (0,.025) -- (0,-.025) node[below] {$0$};
  \draw (1,.025) -- (1,-.025) node[below] {$1$};
  \draw (-.025,1) node[left] {$\nicefrac14$} -- (.025,1);
  \draw (0,0) -- (.94,1.0528);
  \node[overlay,inner sep=.5pt] (slope-label) at (.48,1.18)
    {slope $\Val(\exquant)$};
  \draw[->,overlay] (slope-label.south east) -- (.86,.9632);
  \draw[very thick,domain=0:1,smooth]
    plot (\x,{4*\x*(1-\x)});
  \draw[dashed] (.72,0) -- (.72,.8064) -- (0,.8064);
  \fill (.72,.8064) circle (1.4pt);
  \node[below] at (.72,0) {$\exquant$};
  \node[left] at (0,.8064) {$\rev(\exquant)$};
\end{tikzpicture}

\caption{Revenue curve.}
\label{fig:single-sample-revenue-curve}
\end{subfigure}
\hfill
\begin{subfigure}[t]{0.31\linewidth}
\centering
\begin{tikzpicture}[
  x=2.65cm,
  y=1.43cm,
  >=Stealth,
  every node/.style={font=\scriptsize}
]
  \fill[gray!10] (0,0) rectangle (1,1);
  \fill[gray!24]
    plot[domain=0:1,smooth] (\x,{4*\x*(1-\x)})
    -- (1,0) -- (0,0) -- cycle;
  \fill[gray!45] (0,0) -- (.5,1) -- (1,0) -- cycle;
  \draw[dashed] (0,0) rectangle (1,1);
  \draw[very thick,domain=0:1,smooth]
    plot (\x,{4*\x*(1-\x)});
  \draw[->] (0,0) -- (1.08,0) node[below right] {$\quant$};
  \draw[->] (0,0) -- (0,1.08) node[above] {$\rev(\quant)$};
  \draw (0,.025) -- (0,-.025) node[below] {$0$};
  \draw (1,.025) -- (1,-.025) node[below] {$1$};
  \fill (.5,1) circle (1.4pt);
  \node[below] at (.5,0) {$\monoq$};
  \node[left] at (0,1) {$\rev(\monoq)$};
\end{tikzpicture}

\caption{Single-sample revenue.}
\label{fig:single-sample-revenue-geometry}
\end{subfigure}
\endgroup
   \caption{Single-sample monopoly pricing for $U[0,1]$.
    Panel~\subref{fig:single-sample-inverse-demand} plots inverse demand:
    price $\Val(\exquant)$ sells with probability $\exquant$.
    Panel~\subref{fig:single-sample-revenue-curve} plots the revenue curve
    $\rev(\exquant)=\exquant\,\Val(\exquant)$, which is concave for regular
    distributions.  In panel~\subref{fig:single-sample-revenue-geometry},
    a sample has a uniform quantile, so expected random-reserve revenue is
    the area under $\rev$ (medium gray).  Concavity places beneath $\rev$
    a triangle (dark gray) with half the area of the monopoly-revenue
    rectangle (light gray).  Adapted from \citet{har-13}.}
  \label{fig:single-sample-monopoly-pricing}
\end{figure}

\begin{theorem}[\citealp{DRY-15}]
  \label{thm:single-sample-monopoly-pricing}
  For regular distributions and the revenue objective, single-sample
  monopoly pricing has prior-robust ratio \(\nicefrac12\).
\end{theorem}

\begin{proof}
Let $\Val(\quant) = \dist^{-1}(1-\quant)$ denote inverse demand at
quantile $\quant$, and define the revenue curve by
$\rev(\quant)=\quant\,\Val(\quant)$.  Regularity is equivalent to
concavity of $\rev$.  See
\Cref{fig:single-sample-inverse-demand,fig:single-sample-revenue-curve} and \Cref{def:regular-mhr}.

Write the random reserve as $\critval=\Val(\exquant)$.  Because
$\critval\sim\dist$, its quantile $\exquant$ is uniform on $[0,1]$.
The expected revenue for a given $\exquant$ is $\rev(\exquant)$.  Thus,
the expected random-reserve revenue, depicted by the shaded area under
$\rev$ in \cref{fig:single-sample-revenue-geometry}, is
\[
\expect[\exquant]{\rev(\exquant)} = \int_0^1 \rev(\exquant)\,\dd\exquant.
\]

Let $\monoq$ maximize $\rev$.  The optimal monopoly revenue is
$\rev(\monoq)$, which can be interpreted geometrically as the area of
the smallest rectangle containing $\rev$ in
\cref{fig:single-sample-revenue-geometry}.  Concavity and
nonnegativity imply that $\rev$ lies above the inscribed triangle of
$\rev$ shown there, joining $(0,0)$, $(\monoq,\rev(\monoq))$, and
$(1,0)$.  The area of this inscribed triangle is exactly half the area
of the inscribing rectangle.  Consequently,
\[
\expect[\exquant]{\rev(\exquant)}
  \geq \nicefrac{1}{2}\,\rev(\monoq),
\]
which proves the claim.
It is easy to see that this analysis is tight on the distribution with
revenue curve $\rev(\quant)=1-\quant$.
\end{proof}

Viewing the buyer's value and the random reserve as two i.i.d.\ bids,
symmetry implies that the second-price auction earns twice the revenue of
monopoly pricing with a random reserve.  Thus
\cref{thm:single-sample-monopoly-pricing} is equivalently the two-bidder
Bulow--Klemperer guarantee.

\subsection{Multi-unit Prior-robust Auctions}

Suppose that $k$ identical units are available and each bidder demands
at most one unit.  A welfare-maximizing $k$-unit auction serves the $k$
highest-valued bidders and charges every winner the $(k+1)$st-highest
value.  The \citet{kir-06} proof of the theorem of \citet{BK-96} can
easily be extended to show that $k$ additional agents are sufficient
to guarantee that the welfare-maximizing auction obtains the revenue of the
revenue-optimal auction (without the additional bidders).

\begin{theorem}[\citealp{RTY-12}]
\label{thm:multi-unit-bulow-klemperer}
For $n$ i.i.d.\ regular bidders in a $k$-unit environment, the
welfare-maximizing auction with $n+k$ bidders has expected revenue at least the
Bayesian optimal revenue with the original $n$ bidders.
\end{theorem}

\begin{corollary}[\citealp{RTY-12}]
\label{cor:k-unit-vickrey-ratio}
For $1\leq k<n$, i.i.d.\ regular values, and the revenue objective, the
welfare-maximizing auction has prior-robust ratio at least
$1-\nicefrac{k}{n}$.
\end{corollary}

The guarantee of \Cref{cor:k-unit-vickrey-ratio} is at least one half
when $k\leq \nicefrac n2$, but it approaches zero as $k$ approaches $n$.
\citet{DRY-15} and \citet{RTY-12} give approaches to design mechanisms
that guarantee a large fraction of the optimal revenue: random
reserves and limiting the supply, respectively.

The first construction draws one independent sample from $\dist$ and uses
it as an anonymous reserve in the welfare-maximizing auction.  For i.i.d.\ regular
bidders, the welfare-maximizing auction with the monopoly reserve is Bayesian optimal.
It is enough to compare the sample reserve with this monopoly reserve.

\begin{theorem}[\citealp{DRY-15}]
\label{thm:k-unit-single-sample-reserve}
\label{prop:single-sample-market}
For i.i.d.\ regular values in a $k$-unit environment and the revenue
objective, the welfare-maximizing auction with one independent sample
as an anonymous reserve has prior-robust ratio at least $\nicefrac12$.
\end{theorem}

\begin{proof}[Proof sketch.]
Condition on the values $\valsmi$ of the other agents.  The expected
payments from agent $i$ under the sample-reserve mechanism and the
monopoly-reserve mechanism can then be compared using the same
geometric argument as in
\Cref{thm:single-sample-monopoly-pricing}.  Here the critical quantile
faced by the agent is the minimum of the reserve quantile and the
critical quantile induced by the supply constraint and the other
agents' values.  Denote the latter quantile by
$\exquanti^{\WM}$.  The conditional expected payments under the
sample reserve and the monopoly reserve are, respectively,
\[
  \expect[\exquant]{
    \rev\bigl(\min\{\exquant,\exquanti^{\WM}\}\bigr)
  }
  \quad\text{and}\quad
  \rev\bigl(\min\{\monoq,\exquanti^{\WM}\}\bigr).
\]
A similar inscribed triangle argument can be applied to show the former is at least half of the latter.  Illustrations
can be found, e.g., in \citet{har-13}.  The theorem follows by taking expectation over the other agents' values and summing
across agents.
\end{proof}

While \Cref{thm:k-unit-single-sample-reserve} does not give a fully
prior-robust mechanism, it is straightforward to use the same ideas to
obtain a fully prior-robust mechanism.  Details are in \citet{DRY-15} and \citet{har-13}.  Another approach is by limiting
the supply.

\begin{theorem}[\citealp{RTY-12}]
\label{thm:k-unit-supply-limiting}
For an i.i.d.\ regular $k$-unit environment and the revenue objective,
the welfare-maximizing auction with supply limit $\ell = k$ if $k \leq
\floor{\nicefrac n2}$ and $\ell \sim U\{\floor{\nicefrac n2},\ceiling{\nicefrac n2}\}$,
otherwise, has prior-robust ratio at least $\nicefrac12$.
\end{theorem}

\begin{proof}
  If $k \leq \floor{\nicefrac n2}$ this result follows from
  \Cref{cor:k-unit-vickrey-ratio}.

  Otherwise, the result follows from comparing the revenue of the
  supply-limited mechanism to the optimal revenue with $n$ agents and
  $n$ units.  This latter revenue is $n$ times the monopoly revenue.
  The single-sample theorem
  (\Cref{thm:single-sample-monopoly-pricing}) shows that a uniform
  random quantile reserve gives at least half the revenue.  The threshold
  quantile from the supply constraint $\ell \sim
  U\{\floor{\nicefrac n2},\ceiling{\nicefrac n2}\}$ is a mean-preserving contraction of
  a uniform quantile.  Since $\rev$ is concave, Jensen's inequality
  implies that its expected revenue is at least the random-reserve
  revenue, and hence at least one half of the monopoly revenue.
\end{proof}

\subsection{Optimal Prior-robust Auctions}

The simple mechanisms of the preceding subsections provide robust
guarantees but do not generally obtain the optimal prior-robust
mechanism.  In this subsection we briefly review the literature aiming to
understand the optimal prior-robust mechanism for a single unit and
two agents.  In the theorem below, an only-highest-can-win auction never
allocates to the agent with strictly lower value.  Details can be found
in the original sources.

\begin{theorem}
\label{thm:two-bidder-prior-independent-status}
For two i.i.d.\ bidders and the revenue objective:
\begin{enumerate}
\item Over monotone-hazard-rate distributions, the second-price
  auction is prior-robust optimal and has prior-robust ratio
  $0.715$ \citep{AB-18}.
\item Over regular distributions, no auction has a higher prior-robust
  ratio than $0.556$ \citep{AB-18}.
\item Over regular distributions, among only-highest-can-win auctions,
  the optimal mechanism mixes the second-price auction with one
  relative-markup auction and has prior-robust ratio $0.524$
  \citep{HJL-25}.
\item Over regular distributions, a non-only-highest-can-win auction
  has prior-robust ratio $0.525$ \citep{anu-26}.
\end{enumerate}
\end{theorem}

The optimal prior-robust mechanism for i.i.d.\ regular bidders has not been
identified.

\subsection{Notes}

Prior-robust mechanism design was initiated for digital goods by
\citet{GHW-01,GHKS+06}.  It was connected to the classical economics
approach to revenue maximization of \citet{mye-81} by \citet{seg-03}
and \citet{BV-03}.  Additional early works, largely from the computer
science literature, are surveyed in \citet{HK-07}.
\citet{HR-08,HR-14} connected the early ``prior-free'' framework
formalized in \citet{GHKS+06} to the prior-robust
(\Cref{def:prior-independent-max-min-ratio}) framework discussed
above.  Subsequently, \citet{HJL-20} showed that the frameworks are
equivalent.

\citet{DRY-10,DRY-15} developed the geometric analysis of prior-robust
revenue presented above.  \citet{FILS-15} first demonstrated that for
the $n=2$ agent setting, the second-price auction is not optimal.
They give a mechanism with ratio $0.512$.  \citet{AB-18} showed that
if we further restrict from regular distributions to
monotone-hazard-rate distributions, then the second-price auction is
optimal with ratio $0.715$, and they give improved upper and lower
bounds on the prior-robust optimal auction for regular distributions
as $[0.519,0.556]$.  \citet{AB-18} derive these bounds after reducing
to scale-invariant auctions under a mild regularity condition on
mechanisms; \citet{HJL-25} subsequently establish that restricting
attention to scale-invariant auctions is without loss.  Imposing the
additional constraint on the mechanism to only possibly sell to the
agent with the highest value, \citet{HJL-20,HJL-25} show that a random
markup mechanism has ratio $0.524$.\footnote{The result was originally
stated in \citet{HJL-20} without the restriction on mechanisms;
however, \citet{anu-26} identified a flaw in the argument.}  Recently,
\citet{anu-26} showed that this restriction to selling only to the
highest-valued agent is with loss, and he gives a mechanism with a
slightly better ratio of $0.525$.  It is unknown what the optimal
max-min ratio is and what the prior-robust optimal auction looks like.

\section{Robust Exchange}
\label{s:robust-exchange}

A second domain for robust mechanism design is bilateral exchange.  A
mechanism mediates the exchange of an item between a single buyer and a
single seller.  The Vickrey--Clarke--Groves mechanism achieves
first-best welfare but requires subsidies.  The buyer is charged the
seller's cost (or value for keeping the good), and the seller is paid
the buyer's value.  Since trade only occurs when the buyer's value
exceeds the seller's cost, a positive subsidy is needed.
\citet{MS-83} distinguish between first-best welfare maximization, which
ignores incentive and no-deficit constraints, and second-best welfare
maximization, which imposes those constraints.  They solve the latter
analogously to revenue maximization \citep{mye-81}.

A central question in the literature has been to understand the gap
between second-best and first-best.  When analyzing the fraction of
first-best welfare obtained by the second-best mechanism, it is critical to
choose the zero correctly.  The right zero is the welfare before the
exchange, in which the seller keeps the good and, thus, the zero-adjusted welfare is
simply the gains from trade.

The random proposer mechanism works as follows \citep{BCWZ-17}.
Either the buyer or seller is selected uniformly at random as the
proposer.  The proposer then makes a take-it-or-leave-it offer to the
other agent.  If the offer is accepted, the good is exchanged at the
offered price.  In perfect Bayesian equilibrium, the proposer uses
their prior information to choose an optimal offer, while the principal
needs no knowledge of the distribution of either the buyer or seller.
(This difference from the setup in
\Cref{s:robust-auction-design} will be discussed further in
\Cref{s:revelation-gap}.)  Such mechanisms are prior-robust for the
principal and we will analyze them in the prior-robust max-min ratio
framework (\Cref{def:prior-independent-max-min-ratio}).

We will give the straightforward argument \citep{BDK-21} that the
random proposer mechanism guarantees half of the second-best gains
from trade, i.e., it guarantees a prior-robust ratio of $\nicefrac12$
\citep{BCWZ-17}, and a geometric proof \citep{HW-25} that it obtains a
quarter of the first-best gains from trade \citep{DMSW-22,fei-22}.  A
more involved analysis guarantees a 0.318 fraction
\citep{fei-22,HW-25}.  It is of interest, generally, to understand the
relationship between second-best and first-best separately from
prior-robustness (which compares the prior-robust mechanism to
second-best).  Of course, the second-best mechanism obtains at least as much as
random proposer, which gives an indirect 0.318 guarantee on the max-min ratio of second-best to first-best gains from trade.
\citet{LQRW-26} improve this guarantee to a tight $\nicefrac12$.

\subsection{Random Proposer}

Write \(\val\) for the buyer's value and \(\cost\) for the seller's cost.
The buyer's value and seller's cost are independent, and utilities are linear
in value and transfers.

For a known prior \(\dist\), first-best gains from trade is
$\FB(\dist)=\expect[(\cost,\val) \sim \dist]{\positivepart{\val-\cost}}$.
Second-best gains from trade, denoted by \(\SB(\dist)\),
is the maximum expected gains from trade of a
Bayesian-incentive-compatible interim-individually-rational and
ex-ante-no-deficit mechanism \citep{MS-83}.

\begin{definition}
  \label{def:random-proposer}
The random proposer mechanism ($\RP$) samples the buyer or seller uniformly at random
as the proposer; the proposer makes a take-it-or-leave-it offer to the other
agent, and the other agent accepts the offer if it is individually rational.
\end{definition}

The random proposer mechanism is exactly budget balanced: whenever trade
occurs, the buyer's payment is transferred to the seller.

Without loss, we will describe the outcome when the buyer is the
proposer.  The buyer's problem is the winner-pays-bid best-response
problem analyzed in \Cref{prop:first-price-individual-efficiency} and
depicted in \Cref{fig:individual-efficiency-geometry}.  A buyer with value
\(\val\) proposes the trade price $\bid\in\argmax_{\bid}
  (\val-\bid)\,\prob{\cost\leq\bid}$.

\subsection{Second-best Ratio, I.e., the Prior-robust Analysis}

\begin{theorem}
\label{thm:random-proposer-second-best-half}
For the gains-from-trade objective with independent buyer value and seller
cost, every perfect
Bayesian equilibrium of random proposer obtains at least one half of
second-best gains from trade \citep{BCWZ-17}.  The factor \(\nicefrac12\) is
tight even if the no-deficit requirement is strengthened to the
budget-balanced requirement on the second-best mechanism \citep{BDK-21}.
\end{theorem}

The following simple proof and its tightness is from \citet{BDK-21}, where the
simple proof is attributed to Ilya Segal.

\begin{proof}
The proposer obtains their optimal expected utility over all mechanisms
\citep{mye-81,RZ-83}.  Thus, the
proposer's utility exceeds their utility in the second-best mechanism.  The
gains from trade is the sum of the two agents' utilities.  As each agent is
the proposer half the time, the random proposer mechanism guarantees at least
half the gains from trade of the second-best mechanism.

For tightness, fix the buyer's value at \(\val\) and target buyer
utility $0 < \util < \val$.  Choose the seller's cost distribution so that
its acceptance curve at buyer offer $\bid$ is
\[
  \balloc(\bid)=\frac{\util}{\val-\bid}
\]
until it reaches one at $\bid = \val - \util$.  Every offer on this
curve gives the buyer utility \(\util\), so offering $\bid = 0$ is
optimal.  At this offer the seller has zero
utility and the buyer has the target $\util$; thus, the gains from
trade is \(\util\).  When the seller proposes, the seller offers
\(\val\).  This budget-balanced mechanism trades efficiently, and its gains
from trade equal the area under the acceptance curve:
\[
  \FB(\dist)=\SB(\dist)=\util\left(1+\log\nicefrac{\val}{\util}\right).
\]
Thus, random proposer obtains the fraction
\(\nicefrac12+\frac{1}{2(1+\log(\val/\util))}\) of second-best gains from
trade, which approaches \(\nicefrac12\) as
\(\nicefrac{\util}{\val}\) approaches zero.
\end{proof}

\Cref{thm:random-proposer-second-best-half} shows that the simple and
natural random proposer mechanism guarantees a half of the second-best
gains from trade.  The optimal prior-robust ratio and corresponding
mechanism are unknown; however, this question must be asked carefully
to avoid trivial implementation theory solutions (see further discussion
in \Cref{s:revelation-gap}).

\subsection{First-best Ratio}

The random proposer mechanism guarantees a constant fraction of the
first-best gains from trade as well.  \citet{DMSW-22} show that this
fraction is at least $0.12$.  \citet{fei-22} and \citet{BWZ-26} give lower
bounds of about $0.318$ and \(\nicefrac1e\approx 0.37\), respectively.  \citet{HW-25} give a
simple geometric proof of a slightly worse first-best ratio of
\(\nicefrac14 = 0.25\).  The argument
(\Cref{fig:random-proposer-gft-geometry}) uses the robust-auction
geometry developed in \Cref{s:robust-analysis-auction-equilibria} (cf.~\Cref{fig:individual-efficiency-geometry}).

\begin{figure}
  \centering
  \begingroup
\newcommand{\RPGFTcurve}{(0,0) -- (0.05,0.12) -- (0.10,0.22) -- (0.20,0.35) --
  (0.30,0.38) -- (0.45,0.42) -- (0.60,0.50) --
  (0.75,0.72) -- (0.85,0.85) -- (0.93,0.93) -- (1.00,1.00)
}
\newcommand{\RPGFThalfcurve}{(0,0) -- (0.05,0.06) -- (0.10,0.11) -- (0.20,0.175) --
  (0.30,0.19) -- (0.45,0.21) -- (0.60,0.25) --
  (0.75,0.36) -- (0.85,0.425) -- (0.93,0.465) -- (1.00,0.50)
}
\newcommand{\RPGFTbase}{\draw[->] (0,0) -- (1.08,0) node[below right] {$\bid$};
  \draw[->] (0,0) -- (0,1.08) node[left] {$\quant$};
  \draw[very thick] \RPGFTcurve;
  \node[right] at (1,1) {$\balloc(\cdot)$};
}
\newcommand{\RPGFTguides}{\draw[dashed] (.95,.93) -- (.95,0) node[below] {$\val$};
  \draw[dashed] (.56,.465) -- (.56,0) node[below] {$\bid$};
  \draw[dashed] (.95,.93) -- (0,.93) node[left] {$\balloc(\val)$};
  \draw[dashed] (.56,.465) -- (0,.465) node[left] {$\balloc(\val)/2$};
}

\begin{subfigure}[t]{0.31\linewidth}
\centering
\begin{tikzpicture}[x=2.45cm,y=2.25cm,>=Stealth,every node/.style={font=\scriptsize}]
  \begin{scope}
    \clip (0,0) rectangle (.56,1.08);
    \path[fill=gray!16] \RPGFTcurve -- (1,0) -- (0,0) -- cycle;
  \end{scope}
  \begin{scope}
    \clip (.56,0) rectangle (.95,1.08);
    \path[fill=gray!36] \RPGFTcurve -- (1,0) -- (0,0) -- cycle;
  \end{scope}
  \RPGFTbase
  \RPGFTguides
  \node at (.32,.20) {$\sellerarea$};
  \node at (.74,.33) {$\buyerarea$};
\end{tikzpicture}

\caption{First-best split.}
\label{fig:random-proposer-first-best}
\end{subfigure}
\hfill
\begin{subfigure}[t]{0.31\linewidth}
\centering
\begin{tikzpicture}[x=2.45cm,y=2.25cm,>=Stealth,every node/.style={font=\scriptsize}]
  \path[fill=gray!36] (.56,0) rectangle (.95,.465);
  \RPGFTbase
  \RPGFTguides
  \node at (.75,.23) {$\utilB$};
\end{tikzpicture}

\caption{Buyer-proposer utility.}
\label{fig:random-proposer-buyer-utility}
\end{subfigure}
\hfill
\begin{subfigure}[t]{0.31\linewidth}
\centering
\begin{tikzpicture}[x=2.45cm,y=2.25cm,>=Stealth,every node/.style={font=\scriptsize}]
  \begin{scope}
    \clip (0,0) rectangle (1.05,.465);
    \path[fill=gray!16] \RPGFTcurve -- (1,0) -- (0,0) -- cycle;
    \path[fill=white] \RPGFThalfcurve -- (1,0) -- (0,0) -- cycle;
  \end{scope}
  \RPGFTbase
  \draw[densely dashed, thick] \RPGFThalfcurve;
  \node[right] at (1,.5) {$\balloc(\cdot)/2$};
  \RPGFTguides
  \node at (.34,.26) {$\utilS$};
\end{tikzpicture}

\caption{Seller-proposer utility.}
\label{fig:random-proposer-seller-utility}
\end{subfigure}
\endgroup
   \caption{Acceptance-curve geometry for the random-proposer proof.
  Panel~\subref{fig:random-proposer-first-best} splits first-best gains from
  trade into \(\sellerarea\) and \(\buyerarea\);
  panels~\subref{fig:random-proposer-buyer-utility}
  and~\subref{fig:random-proposer-seller-utility} show the buyer and seller
  deviation utilities.  Adapted from \citet{HW-25}.}
  \label{fig:random-proposer-gft-geometry}
\end{figure}

\begin{theorem}[\citealp{DMSW-22,fei-22,HW-25}]
\label{thm:random-proposer-quarter-approx}
For the gains-from-trade objective with independent buyer and seller values,
every perfect Bayesian equilibrium of random proposer
obtains at least \(\nicefrac 1 4\) of first-best gains from trade.
\end{theorem}

\begin{proof} 
Represent the seller's cost by a random quantile
\(\quant\sim\uniformzeroone\) and a nondecreasing cost function
\(\qcost(\quant)\).  For the geometry, assume that the value and cost
distributions are atomless.\footnote{In fact, this assumption is without loss
due to a construction of \citet{AFHH-13}.}
Fix a buyer value \(\val\), and let
\(\balloc(\bid)=\prob{\qcost(\quant)\leq\bid}=\qcost^{-1}(\bid)\) be the
probability that the seller accepts an offer of \(\bid\).  The conditional
first-best gains from trade is the shaded area under \(\balloc(\cdot)\) up to
\(\val\) in
\Cref{fig:random-proposer-first-best}.  Choose
\(\bid=\qcost(\balloc(\val)/2)\), so that
\(\balloc(\bid)=\balloc(\val)/2\).  This offer divides the first-best area into
the regions \(\sellerarea\) and \(\buyerarea\) shown in
\Cref{fig:random-proposer-first-best}.

When the buyer proposes, offering \(\bid\) gives utility
\((\val-\bid)\balloc(\bid)\), the rectangle in
\Cref{fig:random-proposer-buyer-utility}.  Since
\(\balloc(\bid)=\balloc(\val)/2\), this rectangle covers at least half of
\(\buyerarea\).  The buyer's equilibrium offer is a best response, so
\(\utilB(\val)\geq\nicefrac12\,\buyerarea\).

When the seller proposes, consider the deviation in which type \(\quant\) asks
the price \(\qcost(\min\{2\quant,1\})\).  Since
\(\qcost(\cdot)=\balloc^{-1}(\cdot)\), these offers trace
\(\balloc(\cdot)/2\).  For a fixed \(\quant\), the seller's utility is zero if
this offer exceeds \(\val\); otherwise, it is the horizontal distance between
\(\balloc(\cdot)\) and \(\balloc(\cdot)/2\) in
\Cref{fig:random-proposer-seller-utility}.  Averaging over \(\quant\) gives the
shaded area in that panel.  Its intersection with \(\sellerarea\) is exactly
half of \(\sellerarea\), so the deviation earns at least \(\nicefrac12\,\sellerarea\).
The deviation is independent of the buyer's value \(\val\).
Independence therefore permits the same deviation before
averaging over buyer values, and best response gives
\(\expect{\utilS}\geq\nicefrac12\,\expect{\sellerarea}\).

Each agent proposes with probability \(\nicefrac12\).  The sum of the two agents'
utilities is gains from trade, and the non-proposer's utility is nonnegative.
Consequently,
\[
  \RP(\dist)
    \geq \nicefrac{1}{2}\,\expect{\utilB}+\nicefrac{1}{2}\,\expect{\utilS}
    \geq \nicefrac{1}{4}\,\expect{\buyerarea}
       + \nicefrac{1}{4}\,\expect{\sellerarea}
    =\nicefrac14\,{\FB(\dist)}. \qedhere
\]
\end{proof}

\Cref{thm:random-proposer-quarter-approx} implies that second-best guarantees
a constant fraction of first-best.  \citet{LQRW-26} show that, in fact, the second-best is half of
first-best in worst case.

\begin{theorem}[\citealp{LQRW-26}]
\label{thm:second-best-half-efficient}
For the gains-from-trade objective in bilateral trade, the second-best
mechanism guarantees half the first-best under every independent prior
\(\dist\), i.e., \(\SB(\dist)\geq \nicefrac12\,\FB(\dist)\).
\end{theorem}

\subsection{Notes}

\label{s:robust-exchange-notes}

\citet{MS-83} characterize the incentive-compatible and individually
rational bilateral-trade mechanisms, establish that efficient trade generally
requires an outside subsidy, and formulate the second-best mechanism that
maximizes expected gains from trade subject to incentive, participation, and
no-deficit constraints.

\citet{mca-92} developed a prior-robust
mechanism for multi-agent exchange based on trade reduction.  When there are
\(k\) profitable trades, the trade-reduction mechanism makes at least the
\(k-1\) most profitable of them and thus obtains at least a \(1-\nicefrac1k\) fraction
of first-best gains from trade
(cf.\ \Cref{cor:second-price-worst-case-ratio}).

For the original bilateral-exchange problem, \citet{HR-87} show that the only
dominant-strategy incentive-compatible, ex-post individually rational, and
budget-balanced mechanisms are posted-price mechanisms, in which trade occurs
only when both agents accept the price.  Their formal results cover
differentiable mechanisms, deterministic allocation rules, and arbitrarily
fine step-function rules.  \citet{mca-08}
shows that the first-best ratio of the ex-ante-optimal fixed price is at least the smaller of the
seller and buyer acceptance probabilities at that price.  Consequently, if
the seller's median is no greater than the buyer's median, every price between
the two medians obtains at least one half of first-best gains from trade.

Posted prices face two limitations.  Even under \citeauthor{mca-08}'s
median-ordering assumption, obtaining a constant first-best ratio
requires prior information.  Without this assumption,
no price, even one tailored to the distributions, guarantees a
constant first-best ratio.  \citet{BD-16} give a family of
distributions supported between zero and a growing upper endpoint for
which the first-best ratio of every fixed price is at most a constant
divided by that upper endpoint (cf.\ \Cref{s:robustness-paradigms}).
Subsequently, \citet{CGDL+17} parameterize the achievable guarantee by
the probability that trade is efficient.  They obtain a fixed-price
guarantee that deteriorates only logarithmically as efficient trade
becomes rare, and give matching instances showing that this dependence
is asymptotically tight.  

\citet{BCWZ-17} proved that random proposer obtains at least one half
of second-best gains from trade under the ex-ante no-deficit
constraint; random proposer itself is ex post budget balanced.  The
short proof given above is recorded by \citet{BDK-21} and attributed
to Ilya Segal.  Their tight example is budget balanced, so the
\(\nicefrac12\) ratio remains tight even if second best is required to
be budget balanced.

\citet{DMSW-22} study the first-best ratio of random proposer and
obtain the distribution-independent lower bound of \(0.12\).
\citet{fei-22} improves the first-best ratio to about \(0.318\).
\citet{HW-25} give the geometric proof of the \(\nicefrac14\) ratio
presented above and refine the argument to recover the bound of
\citeauthor{fei-22}.  For buyer and seller distributions with a common
bounded support, \citet{jo-26} improves the lower bound to
\(\nicefrac{1}{\pi} \approx 0.318\) and gives an upper bound of about
\(0.46\).  \citet{BWZ-26} removes the bounded-support restriction and improves
the lower bound to \(\nicefrac1e\approx 0.37\).  Thus, the exact bilateral
ratio remains unknown and lies between approximately \(0.37\) and \(0.46\).

\citet{LQRW-26} prove that second-best gains from trade is always at least
one half of first-best gains from trade and that this factor is tight.
Combined with the second-best ratio of one half from \citet{BCWZ-17}, their
result gives another derivation of the first-best ratio of one quarter.  The
proof of \citet{HW-25} instead compares random proposer directly with first
best.

A separate literature studies prior-robust revenue maximization by an
intermediary.  \citet{BV-03} give an adaptive large-market mechanism that learns
the buyer and seller distributions from the traders' reports and asymptotically
attains the profit of the optimal mechanism with known distributions.
\citet{LM-20} show how to perform the market analysis needed for asymptotic
profit maximization under regularity within a double-clock auction.
For the original bilateral-exchange problem, \citet{LN-07,LN-26} show that the
optimal mechanism for an intermediary can be implemented by a transaction fee
on the seller's posted price.  The optimal fee generally depends on both
distributions.  When seller costs follow a reflected generalized Pareto
distribution, however, the optimal fee is linear and prior-robust with respect
to the buyer's distribution.  They further show that optimal fees converge to
this form in thin markets as the overlap between buyer values and seller costs
vanishes.  \citet{NYK-14} subsequently study affine transaction fees and,
under regularity conditions including linear marginal revenue or monotone
hazard rates, obtain a constant fraction of optimal intermediary profit
uniformly over seller distributions.  Their mechanisms are therefore
prior-robust with respect to the seller's distribution.
 
\section{Revelation Gap}
\label{s:revelation-gap}

Most research on mechanism design begins with the simplification provided by
the revelation principle \citep{mye-81}: if there is a mechanism
with a good equilibrium then there is a mechanism with a good
truthtelling equilibrium.  This changes the optimization over all
mechanisms and all equilibria into a constrained optimization problem:
maximize expected performance subject to incentive compatibility
constraints.  The revelation principle has a simple constructive
proof: the revelation mechanism can simulate the equilibrium in the
non-revelation mechanism and implement the outcome of this simulation.

The revelation principle does not hold for prior-robust mechanism
design.  Here we imagine an ignorant principal who is designing and
running a mechanism for informed agents.  In other words, the aim is
to design one mechanism that performs well in every equilibrium of
every environment.  Given such a scenario, the principal is not free
to simulate the agent strategies as these generally depend on the
prior distribution, and the result of such a simulation would not be
prior robust.

The online marketplace eBay gives an example of this setup.  Their
platform mechanism mediates the exchange between buyers and sellers
for heterogeneous items.  They have a fixed mechanism in which sellers
list items, either by auction or with a posted price, and buyers either
bid or choose whether to accept the posted price.  The eBay mechanism
is not tailored to any specific item for sale.  Instead, eBay
delegates to the seller to configure the mechanism for the buyer and
simply charges a percentage fee for the transaction
\citep[cf.][]{LN-26}.

\citet{FH-18} proposed to study the extent to which non-revelation
mechanisms are more prior-robust than revelation mechanisms.  They
define this deficiency in the prior-robust max-min ratio framework as
the {\em revelation gap}.  The revelation gap is the fraction of the
optimal prior-robust guarantee over all mechanisms that is lost when
attention is restricted to revelation mechanisms.

This section revisits bilateral exchange from \Cref{s:robust-exchange}.
The random proposer mechanism from that section is a non-revelation
mechanism.  In contrast, as we will observe here, a no-deficit
revelation mechanism cannot guarantee any fraction of the second-best
(or first-best) gains from trade.  Thus, the revelation gap is one
(total).  In practice, one-side-proposes mechanisms are common,
e.g., eBay and real estate, while revelation mechanisms are not common
for such exchanges.

We next review the results of \citet{FH-18} on welfare maximization for
agents with budgets and of \citet{FHL-21} on pricing with a single sample.
These results are quite technical, and
so they will be covered only at a high level.

A main challenge in the study of the revelation gap is that some care
is needed in framing the question of prior-robust
non-revelation mechanism design.  Specifically, when considering
prior-robust revelation mechanisms, there is an inherent symmetry:
neither the principal nor the agents need know the prior.  When moving
to non-revelation mechanisms, the asymmetry between a principal seeking
robustness to the prior and agents who know the prior raises the
possibility that the principal could ask directly for the information
it lacks.  This approach is considered in the nonparametric implementation-theory literature
\citep[see][for a survey]{jac-01}.  For example, with unrestricted
message spaces, the mechanism can ask agents to report the prior,
cross-check their reports, and condition its outcome on the reported
distribution. Relevant results include: \citet{CK-99} give a nonparametric undominated Bayesian
implementation for efficient public-good provision using prior
reports, and \citet{CR-05} implement the revenue-maximizing auction
for an uninformed seller.  At the level of our max-min supremum, a
virtual nonparametric implementation corresponds to a prior-robust
max-min ratio of one (i.e., the implementation matches the second-best
performance).  This review article views these positive implementation-theory
results as negative results for the model.

The three examples discussed in this section address this modeling
challenge in different ways.  For bilateral exchange, there is no
non-trivial guarantee for prior-robust revelation mechanisms, so the
details of the prior-robust non-revelation mechanism are not
important.  For welfare maximization with financially constrained
buyers, a natural prior-robust mechanism achieves the second-best
performance, so there is no improvement possible via implementation
theory.  For single-sample pricing, the single-agent restriction limits
report-the-distribution tricks \citep[cf.][]{TTW-26}.

\subsection{Definition}

Fix a performance measure and a family \(\Dists\) of distributions.
Let \(\allM\) be the family of all admissible mechanisms and let
\(\revM\subseteq\allM\) be the family of revelation mechanisms.  In a
non-revelation mechanism, equilibrium strategies may depend on the
prior.  A revelation mechanism instead has agents report only their
payoff types, with truthful reporting forming an equilibrium for every
prior.  All guarantees in this section hold in every equilibrium under
the solution concept stated for the mechanism; accordingly, for a
non-revelation mechanism \(\mech\), \(\mech(\dist)\) denotes its worst
expected performance over those equilibria under \(\dist\).

Recall from \Cref{def:prior-independent-max-min-ratio} that the optimal
prior-robust max-min ratio with mechanism family $\Mechs$ and
distribution family $\Dists$ is
\[
  \OptPIratio(\Mechs,\Dists)
  =
  \sup_{\mech\in\Mechs}\inf_{\dist\in\Dists}
  \frac{\mech(\dist)}{\SB_{\dist}(\dist)}.
\]

\begin{definition}[Revelation gap]
\label{def:revelation-gap}
The revelation gap is the fraction of the
optimal general-mechanism guarantee lost under revelation:
\[
  \RevGap{\Dists}\in[0,1].
\]
If \(\OptPIratio(\allM,\Dists)=0\), the revelation gap is
zero.
\end{definition}

The numerator and denominator use the same objective, distribution
family, and prior-dependent benchmark.  A revelation gap of zero means
that restricting the design space to revelation mechanisms causes no
loss.  A revelation gap of one, a.k.a.~total, means that general
mechanisms have a positive guarantee while revelation mechanisms have
none.

\subsection{Bilateral Exchange, Revisited}

Consider the bilateral-exchange model of \Cref{s:robust-exchange}.
Recall that the seller has cost $\cost$ and the buyer has value $\val$,
drawn independently as \((\cost,\val)\sim\dist\).  Let $\SB_{\dist}$
denote the optimal Bayesian-incentive-compatible
interim-individually-rational ex-ante-no-deficit mechanism.  Recall from
\Cref{thm:random-proposer-second-best-half} that there is a
prior-robust non-revelation mechanism (the random proposer mechanism, \Cref{def:random-proposer})
that guarantees a ratio of $\nicefrac12$ against $\SB_{\dist}(\dist)$.

\begin{corollary}
For no-deficit bilateral exchange and the gains-from-trade objective,
the optimal prior-robust ratio over all mechanisms is
$\OptPIratio(\allM,\Dists) \geq \nicefrac12.$
\end{corollary}

We now show that no revelation mechanism obtains any non-trivial
prior-robust guarantee.  The following proposition follows immediately
by evaluating each constraint on asymmetric point-mass distributions.\footnote{Much of prior-independent mechanism design
instead restricts attention to i.i.d.\ priors.  In that setting, the analogous
reduction requires a more involved argument; see \citet{FH-18}.}

\begin{proposition}
\label{prop:all-prior-ex-post}
A fixed mechanism that satisfies any combination of (a) Bayesian incentive
compatibility, (b) interim individual rationality, and (c) ex-ante no deficit
for every prior $\dist$ in a family $\Dists$ that contains all asymmetric
point-mass distributions also satisfies the corresponding combination of
(a) dominant-strategy incentive compatibility, (b) ex post individual
rationality, and (c) ex post no deficit.
\end{proposition}

\begin{proof}
Fix a type profile and consider the point-mass prior on that profile.  Under
this prior, each interim or ex-ante constraint is the corresponding ex-post
constraint at that profile.  Since the profile was arbitrary, the result
follows.
\end{proof}

The revelation-mechanism impossibility follows by (implicitly) applying
the payment identity of \citet{mye-81}: incentive compatibility and
individual rationality lower bound both agents' utilities, while no
deficit upper bounds their sum.

\begin{proposition}
\label{prop:exchange-no-deficit-zero-ratio}
For a family $\Dists$ that contains all asymmetric point-mass
distributions and the gains-from-trade objective, every
Bayesian-incentive-compatible interim-individually-rational no-deficit
mechanism has prior-robust ratio zero, i.e., $\OptPIratio(\revM,\Dists) = 0$.
\end{proposition}

\begin{proof}
For expositional simplicity, assume that the minimum below is
attained.  The same argument handles the general case by choosing an
approximate minimizer.  Let $\eptalloc$ be the ex post allocation rule, let the minimum gains from trade be:
\[
  \apxratio=\min_{\cost<\val}\eptalloc(\cost,\val),
\]
and let \((\cost,\val)\) be a minimizing profile.  Suppose, for contradiction,
that \(\apxratio>0\).

By the choice of $\apxratio$, every profile with positive gains from
trade has trade probability at least \(\apxratio\).  Dominant-strategy
incentive compatibility and ex post individual rationality require that,
for any \(0<\epsilon<\val-\cost\), by reporting \(\cost+\epsilon\), the
buyer can obtain utility at least
\[
  \apxratio(\val-\cost-\epsilon).
\]
Letting \(\epsilon\) decrease to zero and applying the same argument to
the seller, the buyer's and seller's utilities are bounded as
\[
  u_B(\cost,\val)\ge \apxratio(\val-\cost),
  \qquad
  u_S(\cost,\val)\ge \apxratio(\val-\cost).
\]
Consequently,
\[
  u_B(\cost,\val)+u_S(\cost,\val)\ge 2\apxratio(\val-\cost).
\]

By \Cref{prop:all-prior-ex-post}, ex-ante no deficit for every prior implies
ex post no deficit.  Thus, the agents' combined utility cannot
exceed the gains from trade generated by the mechanism:
\[
2\apxratio(\val-\cost) \leq 
  u_B(\cost,\val)+u_S(\cost,\val)
  \le \eptalloc(\cost,\val)(\val-\cost)
  =\apxratio(\val-\cost).
\]
This is a contradiction since \(\apxratio\) and \(\val-\cost\) are
strictly positive.
\end{proof}

We conclude with the revelation gap result.  The
simple and practical random proposer mechanism obtains a constant
fraction of the optimal gains from trade while no revelation mechanism
obtains any constant fraction.  The following corollary is immediate.

\begin{corollary}
\label{cor:exchange-revelation-gap}
For the independent bilateral-exchange family above, with second-best gains
from trade as the prior-dependent benchmark, the revelation gap is the total
gap of one.
\end{corollary}

In the remainder of the section, we discuss two other mechanism design
environments with nonzero revelation gaps.  For these environments, we
will give the known bounds and describe the proof ideas that
distinguish the prior-robust revelation and all-mechanism design
problems.

\subsection{Public-Budget Welfare}

Consider the provision of an item to financially constrained buyers so
as to maximize welfare.  There is one item and $n$ agents.  Agent $i$
has value $\vali\geq 0$, and all agents have the same public budget
$\budget$.  The agents have linear utility up to their budget.  The
good must be allocated with no positive transfers.  Thus, a feasible outcome
has allocation indicators $\alloci\in\{0,1\}$ satisfying
$\sum_i\alloci\leq1$ and payments $\payi\in[0,\budget]$; agent $i$'s
utility is $\vali\,\alloci-\payi$.

Consider welfare maximization without positive transfers for i.i.d.\
agents whose values are drawn from an unknown public-budget regular
distribution \citep{mas-00}.  For welfare maximization, a distribution is
public-budget regular if its cumulative distribution function is concave.
For this problem, the
Bayesian optimal mechanism $\SB_{\dist}$ can be described as ironing
high-valued agents and efficiently allocating among low-valued
agents.  The ironed interval depends on the distribution and number of
agents.  The all-pay auction solicits bids, chooses the highest bidder
as the winner, and charges all bidders their bids.  The Bayes--Nash
equilibrium of the all-pay auction selects the optimal ironed interval
endogenously.

\begin{theorem}[\citealp{mas-00}]
  \label{thm:all-pay}
For i.i.d.\ public-budget regular agents, the all-pay auction in unique
Bayes--Nash equilibrium obtains the highest welfare among mechanisms
with no positive transfers.
\end{theorem}

A consequence of \Cref{thm:all-pay} is that, for the public-budget
welfare objective, the optimal prior-robust ratio over all mechanisms
(including non-revelation mechanisms) is one.  \citet{FH-18} combine
this result with an upper bound on the prior-robust ratio of revelation
mechanisms to establish a non-trivial revelation gap.

\begin{theorem}[\citealp{FH-18}]
\label{thm:public-budget-revelation-gap-summary}
For the public-budget welfare objective, the optimal prior-robust ratio
over revelation mechanisms satisfies
\[
  \nicefrac1e
  \le \OptPIratio(\revM,\Dists)
  \le 0.99.
\]
Consequently, the public-budget welfare revelation gap lies between \(0.01\)
and \(1-\nicefrac1e\).
\end{theorem}

For the lower bound on revelation mechanisms, the clinching auction
obtains at least a \(\nicefrac1e\) fraction of the Bayesian optimum.  The
proof compares a clinching-with-lotteries allocation to an ex-ante
relaxation.  For the upper bound on revelation mechanisms,
\citet{FH-18} consider two agents with budget one and values uniform
on \([0,5.5]\).  For two agents, \citet{FH-18} show that Bayesian incentive
compatibility and interim individual rationality for every i.i.d.\ prior imply
dominant-strategy incentive compatibility and ex post individual rationality
(cf.\ \Cref{prop:all-prior-ex-post}).
For this distribution, the welfare-optimal
dominant-strategy incentive-compatible mechanism is a middle-ironed clinching
auction.  Its welfare ratio relative to the Bayesian optimum is at most
\(0.99\).

\subsection{Pricing from a Single Sample}

Consider the revenue-maximizing monopolist's pricing problem when the
seller has access to a single sample from the buyer's value
distribution, cf.\ \Cref{thm:single-sample-monopoly-pricing}.  There
is one seller, one buyer, and one item.  The seller first commits to a
mechanism.  The sample \(\sample\) and buyer value \(\val\) are then
drawn independently from the same unknown distribution
\(\dist\in\Dists\).  Before choosing her action, the buyer knows
\(\dist\) and observes her value, but not the realized sample.  The
seller then implements the committed mechanism using the buyer's
action and the sample.  Note that for the second-best mechanism, i.e.,
when the distribution is known, the sample conveys no extra
information and the mechanism simply posts the monopoly price for the
distribution.

Without loss, a revelation mechanism maps the sample into a
distribution over posted prices \citep{ABB-22}.  \citet{ABB-22} derive
the revelation-mechanism bounds below by restricting attention to
scale-invariant mechanisms.  This restriction is without loss for the
robust optimization \citep{HJL-25}.  \citet{FHL-21} construct a
non-revelation mechanism that performs better.  For
monotone-hazard-rate (MHR) distributions (\Cref{def:regular-mhr}),
their mechanism can be described using a binary action space.  For a
multiplier \(\samplemult>0\), the buyer chooses whether to accept or
reject the price $\samplemult\,\sample$ before observing the sample \(\sample\).  If the buyer rejects,
there is no allocation or payment.  If the buyer accepts, she receives
the item and pays \(\samplemult\sample\), which is unknown at the time
of the decision.  These results give the following revelation gap.

\begin{theorem}
\label{thm:mhr-pricing-revelation-gap-summary}
In single-buyer, single-sample pricing for \(\MHR\) distributions and
the revenue objective, the optimal revelation max-min ratio lies in
\citep{ABB-22}:
\[
  [0.644,0.648].
\]
The optimal max-min ratio over all sample-based mechanisms (including
non-revelation mechanisms) lies in \citep{FHL-21}:
\[
  [0.772,0.931].
\]
Consequently, the revelation gap lies in
\[
  \left[1-\nicefrac{0.648}{0.772},
        1-\nicefrac{0.644}{0.931}\right]
  \approx \left[0.16,0.31\right].
\]
\end{theorem}

For the all-mechanism lower bound, the buyer's expected utility from
accepting is \(\val-\samplemult\,\mathbb{E}[\sample]\), so they
accept if and only if \(\val\geq\samplemult\,\mathbb{E}[\sample]\).
The buyer's distribution-dependent equilibrium decision therefore
induces the threshold \(\critval = \samplemult\,\mathbb{E}[\sample]\),
even though the mechanism is fixed.  With \(\samplemult=0.824\), the
mechanism has worst-case ratio at least \(0.772\).

Two steps make the analysis of the non-revelation mechanism
challenging.  On the constructive side, the proof compares
the revenue from the distribution-dependent threshold \(\critval =
\samplemult\,\mathbb{E}[\sample]\) with the monopoly revenue uniformly
over \(\MHR\) distributions and optimizes the multiplier
\(\samplemult\).  On the impossibility side, \citet{FHL-21} compare a
uniform distribution on \([1,2]\) with point-mass distributions at
realized samples.  A buyer from a point-mass instance can imitate the
strategy of a buyer with value \(2\) from the uniform instance.
Incentive compatibility and the two performance guarantees then imply
the upper bound \(0.931\).

The revelation bounds require a separate optimization over rules that
map the sample directly to a price \citep{ABB-22}.

\subsection{Notes}

See \Cref{s:robust-exchange-notes} for an extensive discussion of
prior work on bilateral exchange.  The revelation-gap analysis above
focuses on no-deficit mechanisms and therefore implies the same result
for budget-balanced mechanisms.  The budget-balanced special case also
follows from the randomized-posted-price characterization initiated by
\citet{HR-87} and proved in general by \citet{CP-16}.  Under a randomized
posted-price mechanism, trade occurs if both agents accept the realized
price.  Any fixed distribution over prices assigns arbitrarily little
probability to some profitable interval, implying a prior-robust ratio
of zero for the gains-from-trade objective.

There is a literature on revenue maximization in bilateral exchange
where non-revelation mechanisms do well.  See the discussion of \citet{LN-07,LN-26} and \citet{NYK-14} in \Cref{s:robust-exchange-notes}.

For welfare maximization with public budgets, \citet{mas-00}
characterizes the Bayesian optimal mechanism.  \citet{aus-04}
introduces the ascending clinching auction and \citet{DHH-13} develop
prior-robust methods for welfare and revenue maximization using the
clinching auction.  \citet{FH-18} bring these perspectives together in
their formulation and study of the revelation gap.  An extensive
literature on budget constraints in mechanism design considers other
aspects of the problem; further citations are omitted.

The pricing-from-samples model was developed by \citet{DRY-10,DRY-15}
with the main result that using the random sample as a posted price
guarantees a ratio of $\nicefrac12$ against the optimal monopoly revenue
(\Cref{thm:single-sample-monopoly-pricing}).  \citet{HMR-18} show that
the \(\nicefrac12\) ratio guarantee is optimal among deterministic pricing rules.
\citet{FILS-15} show that the bound can be improved by randomization.
\citet{AB-19} and \citet{ABB-22} give the bounds reviewed above on the optimal randomized
sample-pricing guarantee for regular and \(\MHR\) distributions.  The
result presented above, that non-revelation mechanisms can outperform
revelation mechanisms, is from \citet{FHL-21}.  Recently,
\citet{TTW-26} revisit the model taking the implementation-theory
perspective of asking the agent to report the distribution, while
enforcing the incentive compatibility constraint.

\bibliographystyle{plainnat}

\end{document}